\documentclass[a4paper,11pt, oneside]{article}
\usepackage[labelfont=bf]{caption}

\emergencystretch=\maxdimen
\usepackage{amsmath,amsfonts,amssymb,amsthm,booktabs,color,epsfig,graphicx,hyperref,url}
\usepackage{csquotes,enumerate}
\usepackage[square,numbers,sort&compress]{natbib}

\theoremstyle{plain}
\newtheorem{theorem}{Theorem}
\newtheorem{exa}{Example}

\newtheorem{proposition}{Proposition}
\newtheorem{lemma}{Lemma}
\newtheorem{corollary}{Corollary}

\theoremstyle{definition}
\newtheorem{definition}{Definition}
\newtheorem{remark}{Remark}

\usepackage{subcaption}
\usepackage{dsfont}

\newcommand{\I}{\mathbb{I}} 
\newcommand{\C}{\mathcal{C}} 
\newcommand{\R}{\mathbb{R}} 
\newcommand{\N}{\mathbb{N}} 
\newcommand{\ec}{\hat{C}_n} 
\newcommand{\PP}{\mathbb{P}} 

\newtheorem{conjecture}{Conjecture}
\newtheorem{condition}{Condition}

\allowdisplaybreaks

\usepackage{tikz}
\usetikzlibrary{positioning}
\usepackage{tikz-cd}
\usepackage{tkz-euclide}

\allowdisplaybreaks

\numberwithin{equation}{section}

\begin{document}

\title{A novel positive dependence property and its impact \\ on a popular class of concordance measures}
\author{S. Fuchs\footnote{{Department for Artificial Intelligence \& Human Interfaces, University of Salzburg, Hellbrunnerstrasse 34, 5020 Salzburg, Austria. sebastian.fuchs@plus.ac.at}}, \;
M. Tschimpke\footnote{{Department for Artificial Intelligence \& Human Interfaces, University of Salzburg, Hellbrunnerstrasse 34, 5020 Salzburg, Austria. marco.tschimpke@plus.ac.at}}}

\date{}

\maketitle


\begin{abstract}
\noindent
A novel positive dependence property is introduced, 
called positive measure inducing (PMI for short), 
being fulfilled by numerous copula classes, 
including Gaussian, Fr\'echet, Farlie-Gumbel-Morgenstern and Frank copulas; 
it is conjectured that even all positive quadrant dependent Archimedean copulas meet this property.
From a geometric viewpoint, a PMI copula concentrates more mass near the main diagonal than in the opposite diagonal.
A striking feature of PMI copulas is that they impose an ordering on a certain class of copula-induced measures of concordance, 
the latter originating in \citet{edwards2004measures} and including Spearman's rho $\rho$ and Gini's gamma $\gamma$, 
leading to numerous new inequalities such as $3 \gamma \geq 2 \rho$.
The measures of concordance within this class are estimated using (classical) empirical copulas and the intrinsic construction via empirical checkerboard copulas, and the estimators' asymptotic behaviour is determined.
Building upon the presented inequalities, 
asymptotic tests are constructed having the potential of being used for detecting whether the underlying dependence structure of a given sample is PMI, which in turn can be used for excluding certain copula families from model building.
The excellent performance of the tests is demonstrated in a simulation study and by means of a real-data example.
\end{abstract}

\noindent\textit{Keywords: }Copula; Dependence property; Estimator; Measures of concordance; Test

\section{Introduction\label{sec:1}}

For capturing dependence relationships between (continuous) random variables $X$ and $Y$,
it is quite common to use single value quantities such as the concordance measures Spearman's rho $\rho_{XY}$, Kendall's tau $\tau_{XY}$ and Gini's gamma $\gamma_{XY}$ (see, e.g., \cite{nelsen2007introduction,durante2016principles,genest2010spearman}), 
the various tail dependence coefficients and functions (see, e.g., \cite{joe2014dependence}), or, more recently, measures of predictability (being capable of detecting directed relationships) as presented in \cite{fgwt2020,siburg2013,sfx2022phi,chatterjee2020,ansari2022}.
Another option consists in examining whether certain (positive) dependence properties are fulfilled such as positive quadrant dependence (PQD), left tail decreasingness (LTD), right tail increasingness (RTI), stochastic increasingness (SI), or total positivity of order 2, the latter considered either for a copula, 
its Markov kernel (%
see \cite{fuchs2023total}), or (if existent) its density. 

The two approaches are linked in that specific dependence properties impose a relationship between certain dependence measures.
For example, 
$3\tau_{XY} \geq \rho_{XY} \geq 0$ whenever the random variables $X$ and $Y$ are PQD (see, e.g., \cite{nelsen2007introduction}),
$\rho_{XY} \geq \tau_{XY} \geq 0$ whenever $X$ and $Y$ are LTD \& RTI (see, e.g., \cite{caperaa1993spearman,fredricks2007relationship}),
and $3\tau_{XY} \geq 2\rho_{XY}$ whenever the connecting copula $C_{XY}$ is absolutely continuous and fulfills $C_{XY} \neq \Pi$ as well as $(\partial^2/\partial u \partial v) \ln |C_{XY}(u,v) - \Pi(u,v)| \geq 0$ (see \cite{fredricks2007relationship}) with $\Pi$ denoting the independence copula.
Further inequalities between Spearman's rho and Kendall's tau comprise the well-known universal inequalities 
$$
    |3\tau_{XY} - 2\rho_{XY}| \leq 1
    \qquad \textrm{ and } \qquad
    \frac{(1+\tau_{XY})^2}{2} - 1 \leq \rho_{XY} \leq 1 - \frac{(1-\tau_{XY})^2}{2}
$$
going back to \citet{daniels1950,durbin1951inversions} (see also \cite{schreyer2017exact}), \pagebreak
inequalities concerning extreme order statistics (see, e.g., \cite{chen2007note,kochar2005}), and
the popular Hutchinson-Lai inequalities 
$$
    -1 + \sqrt{1+3\tau_{XY}} \leq \rho_{XY} \leq \min\left\{ 3 \, \frac{\tau_{XY}}{2}, 2 \tau_{XY} - \tau_{XY}^2 \right\}
$$
conjectured to hold for random variables $X$ and $Y$ being SI (see \cite{hutchinson1990}),
but were disproved in \cite{munroe2010,nelsen2007introduction}.
In this context, it was shown in \cite{hurlimann2003hutchinson,trutschnig2018sharp} that the Hutchinson-Lai inequalities and the inequality
$\tfrac{3 \tau_{XY}}{2+\tau_{XY}} \leq \rho_{XY}$
hold for those SI random variables whose connecting copula is extreme value.
\\
In the authors' understanding, Spearman's rho and Kendall's tau have been considered almost exclusively so far, as they exhibit a high compatibility with the above-mentioned dependence properties.
In contrast, Gini's gamma, for example, fulfills $2 \tau_{XY} \geq \gamma_{XY} \geq 0$
whenever $C_{XY} \geq (M+W)/2$ (with $M$ and $W$ denoting the lower and upper Fr\'echet-Hoeffding bounds), the latter condition being in no relation to the properties PQD, LTD and RTI.

Instead of establishing inequalities for dependence measures in the presence of certain dependence properties, in the present paper we choose a different approach and consider a class of measures of concordance $(\kappa_A)_{A \in \mathcal{C}^\ast}$ (going back to \citet{edwards2004measures} and including Spearman's rho and Gini's gamma) that is generated by an ordered set of copulas $(\mathcal{C}^\ast, \preceq)$ and tackle the question: 
What kind of property needs to be fulfilled by a dependence structure for the values of the measures within this class to be ordered?
\\
As a result, in Section \ref{Sec.PMI}, we come up with a novel dependence property, 
called \emph{positive measure inducing} (PMI for short, Definition \ref{Def.PMI}),
that is based on a copula's reflections, fails to be generally linked to the above mentioned dependence properties PQD, LTD, RTI, etc., and is fulfilled by numerous copula classes including
Gaussian copulas, 
Fr\'echet copulas,
Farlie-Gumbel-Morgenstern copulas and 
Frank copulas. 
We even conjecture that all those Archimedean copulas that are PQD meet this property.
Geometrically speaking, a PMI copula concentrates more mass near the main diagonal than in the opposite diagonal. 
\\
In Section \ref{SectionConcordanceMeasureByCopula} we then recapitulate how so-called invariant (with respect to permutations and reflections) copulas $A \in \mathcal{C}^\ast$ are used to construct measures of concordance $\kappa_A$ and resume an ordering $\preceq$ on invariant copulas which then allows to prove that, for each pair of copulas $A$ and $B$ with $A \preceq B$ and each pair of random variables $X$ and $Y$ whose connecting copula $C_{XY}$ is PMI, the corresponding measures of concordance $\kappa_A$ and $\kappa_B$ are ordered, 
leading to numerous new inequalities including, for example, the following inequality involving Spearman's rho and Gini's gamma
\begin{align} \label{Intro.Ineq.GammaRho}
  3 \gamma_{XY} \geq 2 \rho_{XY}\,.
\end{align}
Building upon this main result, in Section \ref{SectionAsymptoticTesting} we then construct asymptotic tests 
that have the potential of being used for detecting whether the underlying dependence structure of a given sample is PMI.
This is of particular interest in practice since, for example in the case of a rejection, certain families of PMI copulas such as 
Gaussian, Fr\'echet, FGM and Frank copulas
(or as we conject those Archimedean copulas that are PQD) may be excluded for model building.
Such asymptotic tests require estimators for the measures of concordance being asymptotically normal. This is achieved in Section \ref{SectionEstimation} where two different approaches for estimating these quantities are employed, one is the construction via (classical) empirical copulas and the other is the intrinsic construction via empirical checkerboard copulas (also known as empirical bilinear copulas); see, e.g., \cite{genest2017asymptotic}.
The excellent performance of the tests is demonstrated in a simulation study and by means of a real-data example. All proofs are deferred to the Appendix \ref{Sec.App}.

\bigskip
Throughout this paper we will write $\I := [0,1]$ and denote by $\lambda$ the Lebesgue measure, be it $1$- or $2$-dimensional.
Bold symbols will be used to denote vectors, e.g., $\mathbf{x}=(x_1,x_2) \in \mathbb{R}^2$.  
\\
We will let $\mathcal{C}$ denote the family of all (bivariate) copulas. 
For every $C \in \mathcal{C}$ the corresponding probability measure will be denoted by $\mu_C$, 
i.e., $\mu_C([0,u] \times [0,v]) = C(u,v)$ for all $(u,v) \in \I^2$;
for more background on copulas and copula measures we refer to \cite{durante2016principles,nelsen2007introduction}. 
For every metric space $(\Delta,\delta)$ the Borel $\sigma$-field on $\Delta$ will be denoted by $\mathcal{B}(\Delta)$.\pagebreak

According to \cite[Theorem 3.4.3]{durante2016principles} and due to disintegration, 
every copula $C$ fulfills
$$
  C(u,v)
	 = \int_{[0,u]} K_C(p,[0,v]) \; \mathrm{d} \lambda(p),
$$
where $K_C$ is (a version of) the Markov kernel of $C$:
A Markov kernel from $\mathbb{I}$ to $\mathcal{B}(\mathbb{I})$ is a mapping 
$K: \mathbb{I}\times\mathcal{B}(\mathbb{I}) \rightarrow \mathbb{I}$ such that for every fixed 
$F\in\mathcal{B}(\mathbb{I})$ the mapping 
$u\mapsto K(u,F)$ is measurable and for every fixed $u \in\mathbb{I}$ the mapping 
$F\mapsto K(u ,F)$ is a probability measure. 
Given a random vector $(U,V)$ with uniformly distributed univariate marginals and a uniformly distributed random variable $V$ on a probability space $(\Omega, \mathcal{A}, \PP)$ 
we say that a Markov kernel $K$ is a regular conditional distribution of $V$ given $U$ if 
$K (U(\omega), F) = \PP( V \in F \,|\, U ) (\omega) $ holds $\PP$-almost surely for every $F\in \mathcal{B}(\mathbb{I})$. 
It is well-known that for each such random vector $(U,V)$ a regular conditional distribution 
$K(.,.)$ of $V$ given $U$ always exists and is unique for $\PP^U$-a.e. $u\in\mathbb{I}$,
where $\PP^U$ denotes the push-forward of $\PP$ under $U$.
For more background on conditional expectation and general disintegration we refer to \cite{kallenberg2002,klenke2008};
for more information on Markov kernels in the context of copulas we refer to 
\cite{durante2016principles, kasper2021weak, sfx2021vine}.

A map $\phi: \C \to \C$ is said to be a \emph{transformation} on $\C$. Let $\Phi$ denote the collection of all transformations on $\C$ and define the \emph{composition} $\circ: \Phi \times \Phi \to \Phi$ by letting $(\phi_1 \circ \phi_2)(C) := \phi_1(\phi_2(C))$. The composition is associative and the transformation $\iota \in \Phi$ given by $\iota(C):=C$ satisfies $\iota \circ \phi = \phi \circ \iota = \phi $ for every $\phi \in \Phi$ and is therefore called the \emph{identity} on $\C$. Thus, $(\Phi, \circ)$ is a semigroup with neutral element $\iota$. The \textit{permutation} $\pi: \C \to \C$  and the \textit{partial reflection} $\nu_1: \C \to \C$ are defined by
$\pi(C)(u,v) := C(v,u)$ and $\nu_1(C)(u,v) := v - C(1-u,v)$
and lead to the transformations $\nu_2 := \pi \circ \nu_1 \circ \pi$ and $\nu := \nu_1 \circ \nu_2$.
Note that the \textit{total reflection} $\nu$ maps a copula $C$ to its survival copula. Let $\Gamma$ denote the smallest subgroup containing $\pi$ and $\nu_1$, i.e. $\Gamma = \{\iota, \nu_1, \nu_2, \nu, \pi, \pi \circ \nu_1, \pi \circ \nu_2, \pi \circ \nu \}$. 
Proofs and further details on the group of transformations may be found in \cite{durante2019reflection,fuchs2014multivariate,fuchs2014bivariate}.  
A copula $C$ is said to be \emph{invariant} (with respect to the group $\Gamma$) if $\gamma(C) = C$ holds for every $\gamma \in \Gamma$.
\begin{exa}[Invariant copulas]\label{Ex.InvariantCop}~~
\begin{enumerate}
    \item The independence copula $\Pi$ is invariant.
    \item For any copula $C$, the arithmetic mean 
      $$
        C_\Gamma := \frac{1}{\vert \Gamma \vert} \sum_{\gamma \in \Gamma} \gamma(C)
      $$
      is invariant.
      In particular, the arithmetic mean $M_\Gamma = (M+W)/2$ of the lower and upper Fr\'echet-Hoeffding bounds $W$ an $M$ is invariant.
      \item The copula $V$ defined by 
        $$
            V(u,v) := 
            \begin{cases}
                M(u,v), &\text{ if } \vert u - v \vert > \frac{1}{2}, \\
                W(u,v), &\text{ if } \vert u + v - 1 \vert > \frac{1}{2}, \\
                \frac{u}{2} + \frac{v}{2} - \frac{1}{4}, &\text{ otherwise}
            \end{cases}
        $$
        is invariant (see, e.g., \cite{nelsen2007introduction}).
\end{enumerate}
Figure \ref{ExampleGammaInvariant} depicts samples of the invariant copulas $\Pi$, $V$ and $M_\Gamma$ introduced above.
\end{exa}

\begin{figure}[!ht]
    \centering
    \begin{subfigure}{.215\textwidth}
        \centering
        \includegraphics[width=1\textwidth]{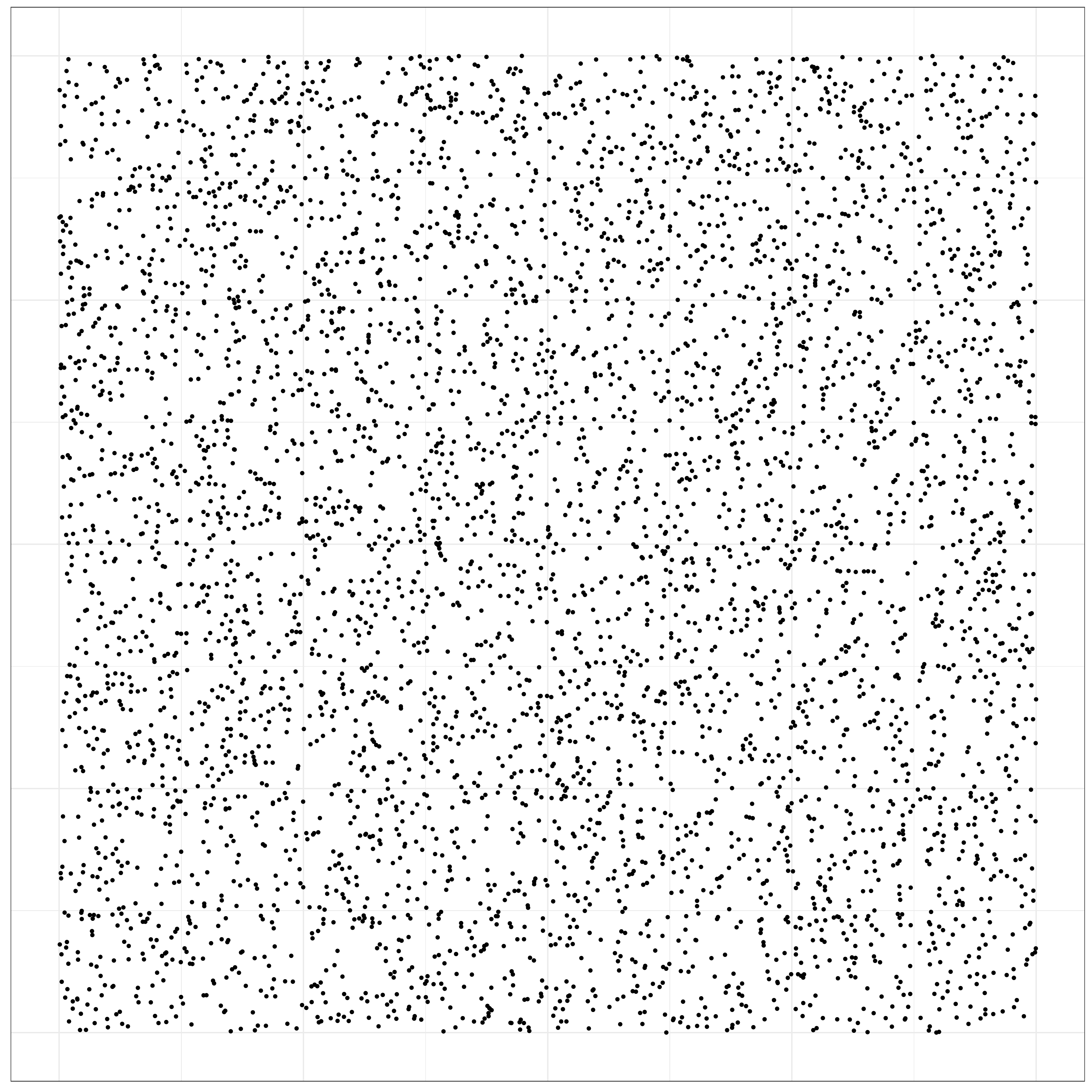}
	    \caption{$\Pi$}
    \end{subfigure}%
    \hspace{0.25cm}
    \begin{subfigure}{.215\textwidth}
        \centering
        \includegraphics[width=1\textwidth]{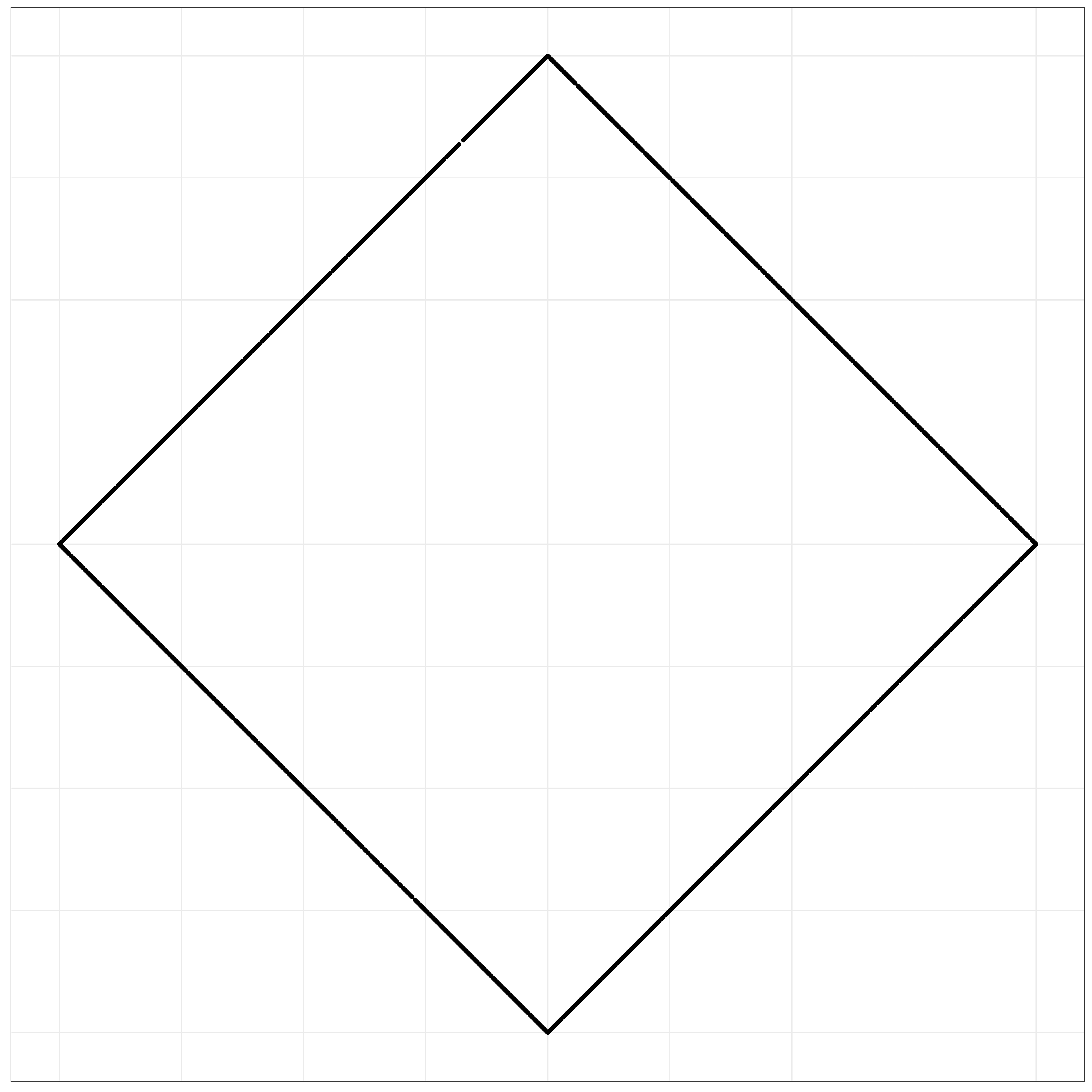}
	    \caption{$V$}
    \end{subfigure}
    \hspace{0.25cm}
    \begin{subfigure}{.215\textwidth}
        \centering
        \includegraphics[width=1\textwidth]{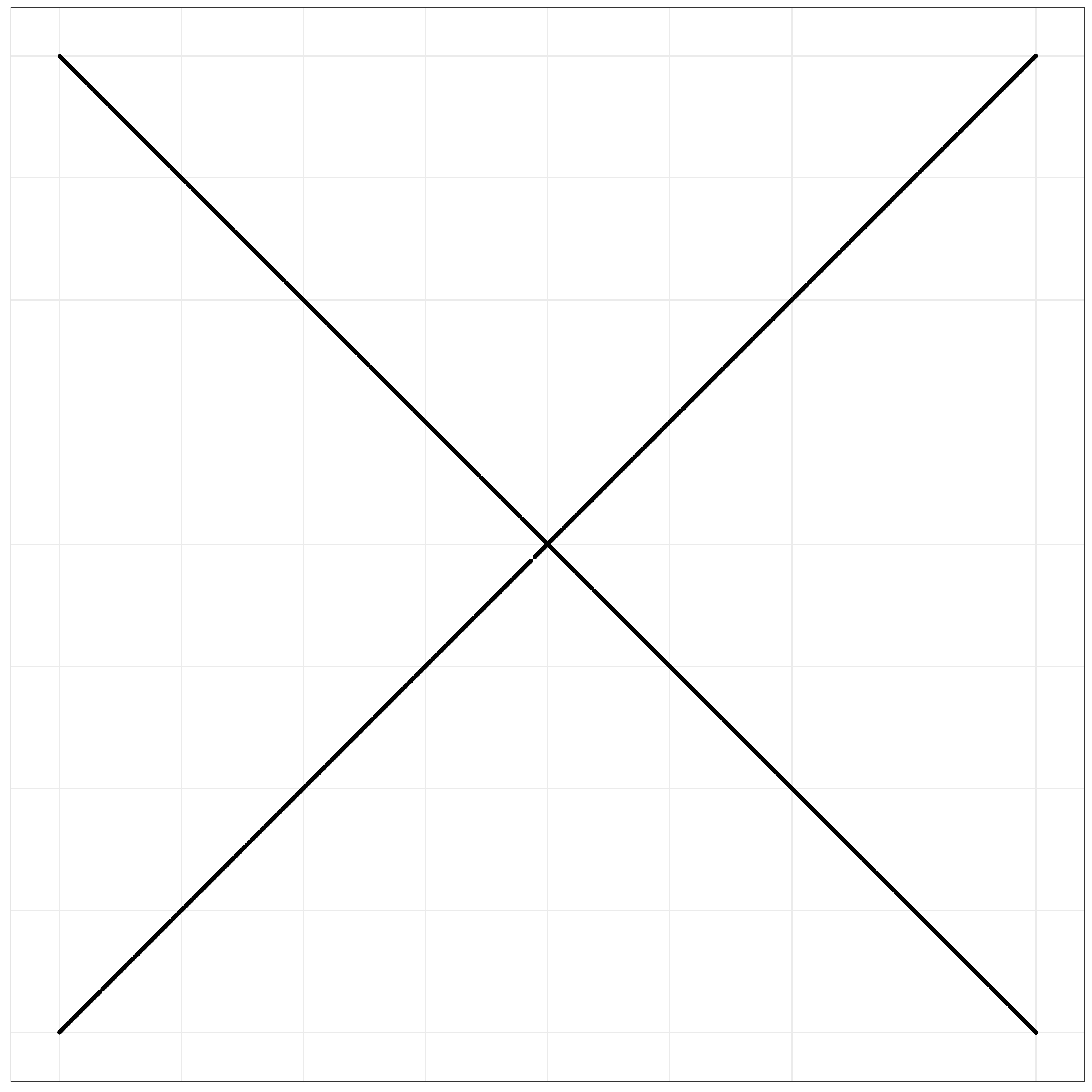}
	    \caption{$M_\Gamma$}
    \end{subfigure}%
    \caption{Samples of size $n = 5.000$ of the invariant copulas $\Pi$, $V$ and $M_\Gamma$ given in Example \ref{Ex.InvariantCop}.} \label{ExampleGammaInvariant}
\end{figure}

\section{The dependence property PMI} \label{Sec.PMI}

In this section a novel dependence property called \emph{positive measure inducing} (PMI for short) is introduced, which is fulfilled by numerous copula classes including
Gaussian copulas, 
Fr\'echet copulas,
Farlie-Gumbel-Morgenstern copulas and 
Frank copulas. 
We even conjecture that all those Archimedean copulas that are positively quadrant dependent (i.e., $C(u,v) \geq \Pi(u,v)$ for all $(u,v) \in (0, 1)^2$) meet this property.
Copulas exhibiting this property concentrate more mass close to the main diagonal than to the opposite diagonal; for an illustration see Figure \ref{IllustraionPMI} below.

\bigskip
Before introducing the novel dependence property, 
we briefly resume some well-known dependence properties for copulas according to 
\cite{fuchs2023total, nelsen2007introduction}: A copula $C\in\C$ is said to be 
\begin{enumerate}[({P}1)]
    \item\label{P1} \textit{positively quadrant dependent (PQD)} if $C(u,v) \geq \Pi(u,v)$ for all $(u,v) \in (0, 1)^2$.
    \item\label{P2} \textit{left tail decreasing (LTD)} if, for any $v \in (0,1)$, the mapping $u \mapsto \frac{C(u,v)}{u}$ is non-increasing, or equivalently, $u \mapsto \frac{\nu_2(C)(u,v)}{u}$ is non-decreasing.
    \item\label{P3} \textit{right tail increasing (RTI)} if, for any $v \in (0,1)$, the mapping $u \mapsto \frac{\nu(C)(u,v)}{u}$ is non-increasing, or equivalently, $u \mapsto \frac{\nu_1(C)(u,v)}{u}$ is non-decreasing.
    \item\label{P4} \textit{stochastically increasing (SI)} if, for (a version of) the Markov kernel $K_C$ and any $v \in (0,1)$, the mapping $u \mapsto K_C(u, [0,v])$ is non-increasing.
\end{enumerate}
For the second type of dependence properties, we recall the notion of a \textit{totally positive of order 2} function 
$f: \Delta \to \mathbb{R}$ with $\Delta \subseteq \I^2$, that is $f$ fulfills
\begin{equation}\label{Prop.TP2}
  f(u_1, v_1) \cdot f(u_2, v_2) - f(u_1, v_2) \cdot f(u_2, v_1) \geq 0
\end{equation}
for all $u_1 \leq u_2$ and all $v_1 \leq v_2$ such that $[u_1,u_2] \times [v_1,v_2] \subseteq \Delta$; see \cite{marshallolkin2011}.
A copula $C\in\C$ is said to be
\begin{enumerate}[({P}1)] \setcounter{enumi}{4}
    \item\label{P5} \textit{TP2} if the copula $C$ is totally positive of order 2 on $(0,1)^2$.
    \item\label{P6} \textit{MK-TP2} if (a version of) the Markov kernel $K_C$ is totally positive of order 2 on $(0,1)^2$.
    \item\label{P7} \textit{d-TP2} if the copula $C$ has a density which is totally positive of order 2 on $(0,1)^2$.
\end{enumerate}
Figure \ref{fig.Implications.Intro} illustrates the relations between the above-mentioned dependence properties,
and a probabilistic interpretation may be found in 
\cite{fuchs2023total,nelsen2007introduction}.

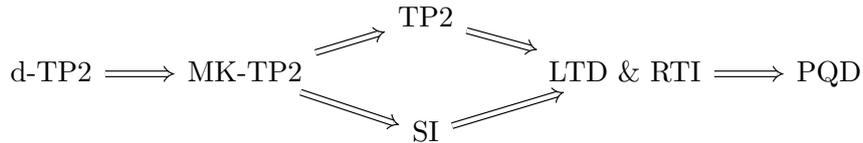
\begin{figure}[!ht]
\centering
  \begin{tikzcd}[row sep=tiny]
            && \textrm{TP2} \arrow[dr, Rightarrow] &\\
            \textrm{d-TP2} \arrow[r, Rightarrow] & \textrm{MK-TP2} \arrow[ur, Rightarrow] \arrow[dr, Rightarrow] & & \textrm{LTD \& RTI} \arrow[r, Rightarrow] & \textrm{PQD} \\
            && \textrm{SI} \arrow[ur, Rightarrow] &\\
  \end{tikzcd}
\caption{Relations between the different notions of positive dependence.}
\label{fig.Implications.Intro}
\end{figure}

We introduce a novel dependence property that is based on a copula's reflections and, for a copula $C$, is formulated in terms of the mapping $E_C: \I^2 \to \R $ defined by 
\begin{align}\label{def.EC}
   E_C := C - \nu_1(C) - \nu_2(C) + \nu(C)
\end{align}

\begin{definition}[Positive measure inducing copulas]\label{Def.PMI}~~\\
A copula $C\in\C$ is said to be 
\begin{enumerate}
\item \textit{positive measure inducing (PMI)} if $E_C$ induces a measure on $ (0,\frac{1}{2})^2$.  
\item \textit{negative measure inducing (NMI)} if $-E_C$ induces a measure on $ (0,\frac{1}{2})^2$.  
\end{enumerate}
\end{definition}

\noindent 
Notice that $E_C$ induces a measure on $(0,\frac{1}{2})^2$ 
if and only if $E_C$ induces a measure on $ (0,\frac{1}{2})^2 \cup (\frac{1}{2},1)^2$ 
if and only if $-E_C$ induces a measure on $ (0,\frac{1}{2}) \times (\frac{1}{2},1) \cup (\frac{1}{2},1) \times (0,\frac{1}{2})$.
In the present paper we mainly focus on property PMI and note here that analogous results can be obtained for the property NMI.

\begin{remark} 
The property PMI can alternatively be formulated either in terms of a copula's Markov kernel or (if existent) in terms of a copula's Lebesgue density:
\begin{enumerate}
\item
A copula $C$ is PMI if and only if, for every $u \in (0,\frac{1}{2})$, the mapping
\begin{align} \label{PMI.MK}
  v \mapsto K_C(u,[0,v]) - K_C(1-u,[0,v]) + K_C(u,[0,1-v]) - K_C(1-u,[0,1-v])
\end{align}
is non-decreasing on $(0,\frac{1}{2})$. An equivalent property applies in the case of interchanged arguments.

\item
An absolutely continuous copula $C$ with Lebesgue density $c$ is PMI if and only if the inequality
\begin{align}\label{PMI.density}
  c(u,v) - c(1-u,v) - c(u,1-v) + c(1-u,1-v) \geq 0
\end{align}
holds for all $(u,v) \in (0,\frac{1}{2})^2$; see left-hand side of Figure \ref{IllustraionPMI} for an illustration.
Obviously, if the density is $2$-increasing, then $C$ is PMI.
\end{enumerate}
Ineq. \eqref{PMI.density} indicates that a PMI copula concentrates more mass near the main diagonal than in the opposite diagonal.
If $C$ is even symmetric, $C$ being PMI is equivalent to
$c(u,v) + c(1-v,1-u) \geq c(1-u,v) + c(1-v,u,)$ for all $(u,v) \in (0,\frac{1}{2})^2$, meaning that $C$ concentrates more mass on parallel lines that are closer to the main diagonal; see right-hand side of Figure \ref{IllustraionPMI} for an illustration.
\end{remark}

\begin{figure}[!ht]
    \centering
    \begin{tikzpicture}[scale=3.25]
	        \draw[-,line width=1] (0,0) -- (1,0);
	        \draw[-,line width=1] (0,0) -- (0,1);
	        \draw[-,line width=1] (1,0) -- (1,1);
	        \draw[-,line width=1] (0,1) -- (1,1);
	        \node[below=1pt of {(-0.04,0.02)}, scale= 0.75, outer sep=2pt] {$0$};
	        \node[below=1pt of {(1,0.02)}, scale= 0.75, outer sep=2pt] {$1$};
	        \node[below=1pt of {(-0.04,1.05)}, scale= 0.75, outer sep=2pt] {$1$};
                \draw[-,line width=1, dotted] (0.5,0) -- (0.5,1);
                \draw[-,line width=1, dotted] (0,0.5) -- (1,0.5);
                \draw[-,line width=1, dotted] (0,0) -- (1,1);
                \draw[-,line width=1] (0.3,0) -- (0.3,-0.02);
	        \node[below=1pt of {(0.3,-0.02)}, scale= 0.75, outer sep=2pt,fill=white] {$u$};
                \draw[-,line width=1] (0.7,0) -- (0.7,-0.02);
	        \node[below=1pt of {(0.7,-0.02)}, scale= 0.75, outer sep=2pt,fill=white] {$1-u$};
                \draw[-,line width=1] (0,0.15) -- (-0.02,0.15);
	        \node[left=1pt of {(-0.02,0.15)}, scale= 0.75, outer sep=2pt,fill=white] {$v$};
                \draw[-,line width=1] (0,0.85) -- (-0.02,0.85);
	        \node[left=1pt of {(-0.02,0.85)}, scale= 0.75, outer sep=2pt,fill=white] {$1-v$};
                \node at (0.3,0.15)[circle,fill,inner sep=1.5pt, blue]{};
                \node at (0.7,0.85)[circle,fill,inner sep=1.5pt, blue]{};
                \node at (0.7,0.15)[circle,fill,inner sep=1.5pt, red]{};
                \node at (0.3,0.85)[circle,fill,inner sep=1.5pt, red]{};
                
                \draw[->] (1.1,1/2) -- (1.9, 1/2) ;
                \node[above=1pt of {(1.5,1/2)}, scale= 0.75, outer sep=2pt,fill=white] {\footnotesize Mirroring on main diagonal};
                
	        \draw[-,line width=1] (2,0) -- (3,0);
	        \draw[-,line width=1] (2,0) -- (2,1);
	        \draw[-,line width=1] (3,0) -- (3,1);
	        \draw[-,line width=1] (2,1) -- (3,1);
	        \node[below=1pt of {(1.96,0.02)}, scale= 0.75, outer sep=2pt] {$0$};
	        \node[below=1pt of {(3,0.02)}, scale= 0.75, outer sep=2pt] {$1$};
	        \node[below=1pt of {(1.96,1.05)}, scale= 0.75, outer sep=2pt] {$1$};
                \draw[-,line width=1, dotted] (2.5,0) -- (2.5,1);
                \draw[-,line width=1, dotted] (2,0.5) -- (3,0.5);
                \draw[-,line width=1, dotted] (2,0) -- (3,1);
                \draw[-,line width=1] (2.3,0) -- (2.3,-0.02);
	        \node[below=1pt of {(2.3,-0.02)}, scale= 0.75, outer sep=2pt,fill=white] {$u$};
                \draw[-,line width=1] (2.7,0) -- (2.7,-0.02);
	        \node[below=1pt of {(2.7,-0.02)}, scale= 0.75, outer sep=2pt,fill=white] {$1-u$};
                \draw[-,line width=1] (2,0.15) -- (1.98,0.15);
	        \node[left=1pt of {(1.98,0.15)}, scale= 0.75, outer sep=2pt,fill=white] {$v$};
                \node at (2.3,0.15)[circle,fill,inner sep=1.5pt, blue]{};
                \node at (2.85, 0.7)[circle,fill,inner sep=1.5pt, blue]{};
                \node at (2.7,0.15)[circle,fill,inner sep=1.5pt, red]{};
                \node at (2.85,0.3)[circle,fill,inner sep=1.5pt, red]{};
	    \end{tikzpicture}
    \caption{Illustration of the novel dependence property PMI for general copulas (left figure) and for symmetric copulas (right figure).}
    \label{IllustraionPMI}
\end{figure}
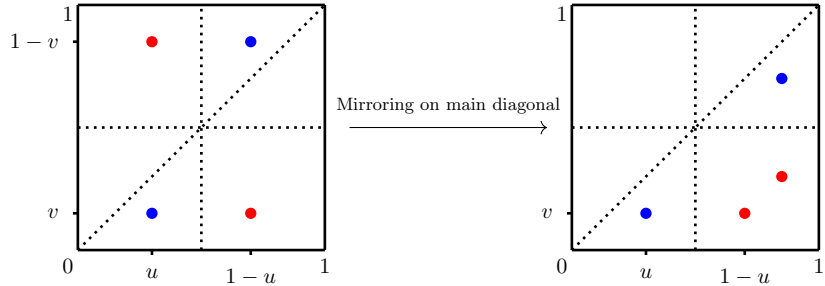

Before discussing examples of PMI copulas we relate the mapping $E_C$ to the dependence properties (P\ref{P1}) - (P\ref{P7}):

\begin{remark} \label{Rem.RelationsP1P4}
\begin{enumerate}
    \item If $C$ is PQD, then $E_C \geq 0$.
    \item IF $C$ is LTD and RTI, then, for any $v \in (0,1)$, the mapping $u \to \frac{E_C(u,v)}{u}$ is non-increasing.
    \item IF $C$ is SI, then, for (a version of) the Markov kernel $K_C$ and any $v \in (0,1)$, the mapping 
    \begin{align} \label{EC.SI}
        u \mapsto K_C(u,[0,v]) - K_C(1-u,[0,v]) + K_C(u,[0,1-v]) - K_C(1-u,[0,1-v])
    \end{align} 
    is non-increasing and $u \mapsto E_C(u,v)$ is concave.
\end{enumerate}
\noindent 
Comparing \eqref{PMI.MK} and \eqref{EC.SI}, it immediately becomes apparent that the properties PMI and SI are not related in general (also see Example \ref{ExampleEVC} below).
Due to the additive structure of $E_C$, similar statements regarding the properties (P\ref{P5}) - (P\ref{P7}) are not to be expected.
\end{remark}

From Examples \ref{Ex.Frechet}, \ref{Ex.FGM} and \ref{ExampleEVC} below we observe that neither PMI implies PQD, 
nor MK-TP2 
implies PMI.
Within certain copula families, however, the equivalence of PMI and the dependence properties (P\ref{P1}) - (P\ref{P7}) can be verified.

\begin{exa}[Invariant copulas and Fr\'echet-Hoeffding bounds]~~
\begin{enumerate}
    \item Every invariant copula $C$ fulfills $E_C = 0$ and hence is PMI (and NMI).
    In particular, $\Pi$, $V$ and $M_\Gamma$ are PMI; compare Example \ref{Ex.InvariantCop}. 

    \item The upper Fr\'echet-Hoeffding bound $M$ fulfills $E_M = 2 (M-W)$ and is PMI. \\
    Indeed, since for every $u \in (0,\frac{1}{2})$, the mapping
    \begin{align*}
      v 
      & \mapsto K_M(u,[0,v]) - K_M(1-u,[0,v]) + K_M(u,[0,1-v]) - K_M(1-u,[0,1-v])
      \\
      & = \mathds{1}_{[u,1]}(v) - \mathds{1}_{[1-u,1]}(v) + \mathds{1}_{[0,1-u]}(v) - \mathds{1}_{[0,u]}(v)
      \\
      & = \mathds{1}_{[u,1-u)}(v) + \mathds{1}_{(u,1-u]}(v)  
    \end{align*}
    is non-decreasing on $(0,\frac{1}{2})$, we immediately obtain that $M$ is PMI.

    \item The lower Fr\'echet-Hoeffding bound $W$ fulfills $E_W = 2 (W-M) = - E_M$ and is not PMI (but NMI). 
\end{enumerate}
\end{exa}

\begin{exa}[Fr{\'e}chet copulas]\label{Ex.Frechet}~~\\
For $\alpha, \beta \in \I$ with $\alpha+\beta \leq 1$, the mapping $C_{\alpha,\beta}: \I^2 \to \I$ given by
\begin{align*}
  C_{\alpha,\beta} (u,v)
  = \alpha \, M(u,v) + (1-\alpha -\beta) \, \Pi(u,v) + \beta \, W(u,v)
\end{align*}
is a copula and called Fr{\'e}chet copula.
$C_{\alpha,\beta}$ fulfills 
$E_{C_{\alpha,\beta}} = 2 \, (\alpha-\beta) (M-W) = (\alpha-\beta) \, E_M$ so that the following statement can be easily verified:
\begin{enumerate}
\item $C_{\alpha,\beta}$ is PMI if and only if $\alpha \geq \beta$.

\item $C_{\alpha,\beta}$ is NMI if and only if $\alpha \leq \beta$.
\end{enumerate}
\noindent By recalling that a Fr{\'e}chet copula $C_{\alpha,\beta}$ is PQD if and only if $\beta=0$,
it immediately becomes apparent that PMI neither implies PQD nor LTD / RTI / SI / TP2 / MK-TP2. 
\end{exa}

We now show that FGM copulas with a non-negative parameter are PMI.
Notice that FGM copulas are one of the very few copulas whose density is even $2$-increasing.

\begin{exa}[Generalized FGM copulas]\label{Ex.FGM}~~\\
    For differentiable functions $f,g: \I \to \R$ with $f(0) = f(1) = g(0) = g(1) = 0$ and
    $f'(u)g'(v) \geq -1 $ for all $u,v \in\I$,
    the function $C_{f,g}: \I^2 \to \I$ given by
    \begin{align}\label{ParametricExample}
        C_{f,g} (u,v) := uv + f(u)g(v)
    \end{align}
    defines a copula with density $c_{f,g}(u,v) = 1 + f'(u)g'(v)$ and fulfills 
    $E_{C_{f,g}} (u,v) = [f(u) + f(1-u)] [g(v) + g(1-v)]$.
    According to Ineq. \eqref{PMI.density}, 
    $C_{f,g}$ is PMI if and only if
    \begin{align*}
      [f'(u) - f'(1-u)] [g'(v) - g'(1-v)] \geq 0
    \end{align*}
    for all $(u,v) \in (0,\tfrac{1}{2})^2$, 
    which is the case if $f,g$ are either both convex or both concave. 
    In particular,
    \begin{enumerate}
    \item
    if, for $\alpha \in [-1,1]$, $f(u) := \alpha \, u(1-u)$ and $g(v) = v(1-v)$, 
    then $C_{f,g}$ coincides with the usual \emph{Farlie-Gumbel-Morgenstern (FGM) copula} and the following statements are equivalent:
    \begin{enumerate}[(a)]
        \item $\alpha \in [0,1]$.
        \item $C_{f,g}$ is PQD / LTD / RTI / SI / TP2 / MK-TP2 / d-TP2.
        \item $C_{f,g}$ is PMI.
    \end{enumerate}
    We note in passing that $C_{f,g}$ is NMI if and only if $\alpha \in [-1,0]$.\pagebreak
    
    \item 
    if, for $\alpha \in [-1,1]$, $f(u) := \alpha \, u^2(1-u)^2$ and $g(v) = v(1-v)$, 
    then $C_{f,g}$ is PMI if and only if $\alpha \in [0,1]$ and $C_{f,g}$ is NMI if and only if $\alpha \in [-1,0]$.
    \item 
    if $f(u) := u(1-u)(1-2u)$ and $g(v) = v(1-v)$, 
    then $C_{f,g}$ is PMI but fails to be PQD / LTD / RTI / SI / TP2 / MK-TP2 / d-TP2.
    \end{enumerate}
\end{exa}

In contrast to FGM copulas, Gaussian copulas and Frank copulas fail to have $2$-increasing densities, in general.
However, taking advantage of their geometric structure, below we prove that Gaussian and Frank copulas are in fact PMI and we make a conjecture for general Archimedean copulas; 
proofs of Examples \ref{Cor.Gaussian} and \ref{ExampleFrank} are deferred to the Appendix \ref{Sec.App}.

\begin{exa}[Gaussian copulas]\label{Cor.Gaussian}~~\\
For $\rho \in (-1,0) \cup (0,1)$, the mapping $C_\rho: \I^2 \to \I$ given by
\begin{align*}
    C_\rho(u,v) 
    := \int\limits_{(-\infty, \Phi^{-1}(u)] \times (-\infty, \Phi^{-1}(v)]}
    \frac{1}{2\pi \sqrt{1-\rho^2}} \; \exp{\left(-\frac{s^2-2\rho st + t^2}{2(1-\rho^2)} \right)} 
    \; \mathrm{d} \lambda(s,t)
\end{align*}
is a copula and called Gaussian copula,
where $\Phi^{-1}$ denotes the inverse of the standard normal distribution function.
Then the following statements are equivalent:
\begin{enumerate}[(a)]
  \item $\rho \in (0,1)$.
  \item $C_\rho$ is PQD / LTD / RTI / SI / TP2 / MK-TP2 / d-TP2.
  \item $C_\rho$ is PMI.
\end{enumerate}
We note in passing that a Gaussian copula is NMI if and only if $\rho \in (-1,0)$.
\end{exa}

In the sequel we study Archimedean copulas (see, e.g., \cite{durante2016principles,nelsen2007introduction}). 
A convex, strictly decreasing function $\varphi: \I \to [0,\infty]$ with $\varphi(1)=0$ is called a generator.
According to \cite{kasper2021weak} we may assume that all generators are right-continuous at $0$. Every generator $\varphi$ induces a symmetric copula C via
$$
    C(u,v) = \psi (\varphi(u) + \varphi(v))
$$
for all $(u,v) \in \I^2$ where 
$\psi: [0,\infty] \to \I$ denotes the pseudo-inverse of $\varphi$ defined by
$$
  \psi(x)
	:= \begin{cases}
     \varphi^{-1}(x) & \text{if } x\in[0,\varphi(0)) \\
     0              & \text{if } x\geq\varphi(0).
     \end{cases}
$$
The copula $C$ is called \emph{Archimedean copula}.


\begin{exa}[Frank copulas]\label{ExampleFrank}\textcolor{white}{a}\\
For $\delta \in (-\infty,0) \cup (0,\infty)$, the mapping $C_\delta: \I^2 \to \I$ given by
\begin{align*}
  C_\delta(u,v) 
  & := - \delta^{-1} \log \left( \frac{1 - e^{-\delta} - (1 - e^{-\delta u}) (1 - e^{-\delta v}) }{1 - e^{-\delta}} \right)
\end{align*}
is a copula and called Frank copula.
The following statements are equivalent:
\begin{enumerate}[(a)]
  \item $\delta \in (0,\infty)$.
  \item $C_\delta$ is PQD / LTD / RTI / SI / TP2 / MK-TP2 / d-TP2.
  \item $C_\delta$ is PMI.
\end{enumerate}
We note in passing that a Frank copulas is NMI if and only if $\delta \in (-\infty,0)$.
\end{exa}

While studying the class of Archimedean copulas we came to the following conjecture,
which we were not yet able to prove:

\begin{conjecture}[Archimedean copulas]~~\\
For Archimedean copulas $C$ we conjecture that the following statements are equivalent:
\begin{enumerate}
  \item[(a)] $C$ is PQD.
  \item[(b)] $C$ is PMI. 
\end{enumerate}
We further conjecture that an Archimedean copulas $C$ is NMI if and only if $C(u,v) \leq \Pi(u,v)$ holds for all $(u,v) \in (0,1)^2$.
\end{conjecture}



A copula $C_A$ is called \textit{extreme value copula} (EVC) if there exists a convex function $A : \I \to \I$ fulfilling \linebreak $\max\{1-t, t\} \leq A(t) \leq 1$ for all $t\in\I$ such that
\begin{align*}
    C_A(u,v) = (uv)^{A\big(\frac{\log(u)}{\log(uv)}\big)}
\end{align*}
for all $u,v \in (0,1)^2$ (see, e.g., \cite{GS2011, de1977limit, Pickands1981}). 
It is well-known that extreme value copulas are SI 
(see, e.g., \cite{guillem2000structure, joe2014dependence}).
The following example verifies that extreme value copulas fail to be PMI, in general.

\begin{exa}[Extreme value copulas]\label{ExampleEVC}~~
\begin{enumerate}
\item
Consider the Pickands dependence function $A$ given by
$$
    A(t) := \begin{cases}
	        1-t         & t \in [0, 0.2)	\\
			0.9 - 0.5 t & t \in [0.2, 0.5) \\
            0.4 + 0.5 t & t \in [0.5, 0.8) \\
            t           & t \in [0.8, 1).
		 \end{cases}
$$
Then $C_A$ is SI, hence 
$u \mapsto K_{C_A}(u,[0,v]) - K_{C_A}(1-u,[0,v]) + K_{C_A}(u,[0,1-v]) - K_{C_A}(1-u,[0,1-v])$  
is non-increasing due to \eqref{EC.SI}.
However, the mapping 
$v \mapsto K_{C_A}(u,[0,v]) - K_{C_A}(1-u,[0,v]) + K_{C_A}(u,[0,1-v]) - K_{C_A}(1-u,[0,1-v])$ 
fails to be non-decreasing at $(u,v_1)=(0.25,0.3)$ and $(u,v_2)=(0.25,0.4)$, 
such that, due to property \eqref{PMI.MK}, $C_A$ fails to be PMI.

\item 
For $\alpha, \beta \in \I$, the mapping $M_{\alpha,\beta}: \I^2 \to \I$ 
given by
$$
    M_{\alpha,\beta} 
    = \begin{cases}
	        u^{1-\alpha}\,v & u^\alpha \geq v^\beta \\
			u\,v^{1-\beta}  & u^\alpha < v^\beta
		 \end{cases}
$$
is a copula and called Marshall-Olkin copula.
According to \cite{fuchs2023total}, $M_{\alpha,\beta}$ is MK-TP2 if and only if $\beta=1$.
However, for $\alpha=0.05$ and $\beta=1$, $M_{\alpha,\beta}$ fails to be PMI
which can be easily verified by evaluating property \eqref{PMI.MK} at the points $(u,v_1)=(0.1,0.1)$ and $(u,v_2)=(0.1,0.15)$.
\end{enumerate}
\end{exa}

\section{The property PMI and its impact on a class of measures of concordance generated by invariant copulas}\label{SectionConcordanceMeasureByCopula}

The present section first recapitulates how invariant copulas are used to construct measures of concordance according to \citet{edwards2004measures},
then resumes an ordering $\preceq$ on invariant copulas that, finally, allows for an ordering of the induced measures of concordance (main Theorem \ref{MainResultPMI.Thm}).

\begin{definition}[Measures of concordance]\label{Def.MOC}~~\\
A map $\kappa: \C \to \mathbb{R}$ is said to be a measure of concordance if it has the following properties:
\begin{enumerate}
    \item[(i)] $\kappa(M) = 1$.
    \item[(ii)] $\kappa(\pi(C)) = \kappa(C)$ for all $C\in\C$.
    \item[(iii)] $\kappa(\nu_1(C)) = -\kappa(C)$ for all $C\in\C$.
    \item[(iv)] $\kappa(C) \leq \kappa(D)$ whenever $C\leq D$,
    i.e., whenever $C$ and $D$ are ordered pointwise.
    \item[(v)] $\lim_{n \to \infty} \kappa (C_n) = \kappa(C)$ for any sequence $\{C_n\}_{n \in \N} \subseteq \C$ and any copula $C\in\C$ such that $\lim_{n \to \infty} C_n = C$ pointwise.
\end{enumerate}
\end{definition}
\noindent Definition \ref{Def.MOC} is in accordance with \citet{sca1984} and
\cite{edwards2005measures,edwards2004measures,edwards2009,fuchs2016copula,fuchs2014bivariate}.
Here, property (iv) 
ensures that the inequality $-1 = \kappa(W) \leq \kappa(C) \leq \kappa(M) = 1$ holds for all $C\in\C$. Proofs and further details on measures of concordance (as defined above) may be found in \cite{fuchs2014bivariate} ($d=2$) and in \cite{fuchs2016copula}.

In the sequel we examine how to construct measures of concordance based on invariant copulas. Therefore, consider the map $[.,.]: \C \times \C \to \R$ given by
\begin{align}\label{Def.BiconvexForm}
    [C,D] := \int_{\I^2} C(u,v) \, \mathrm{d} \mu_D (u,v)\,.
\end{align}
The map $[.,.]$ is in either argument linear with respect to convex combinations and is therefore called a \emph{biconvex form}. 
Moreover, the map $[.,.]$ is symmetric, in either argument monotonically increasing with respect to the pointwise (or concordance) order $\leq$ and satisfies $0 \leq [C,D] \leq [M,M] = 1/2$ 
(see, e.g., \cite{fuchs2016biconvex}).  
\\
Now, consider a fixed invariant copula $A$. 
Then $[M,A] > [\Pi,A]$ so that the map $\kappa_A: \C \to \R$
given by
\begin{align*} 
    \kappa_A(C) := \frac{[C,A] - [\Pi, A]}{[M,A] - [\Pi, A]}
\end{align*}
is well-defined; see \cite{fuchs2016copula}. 
The following result by \citet{edwards2004measures} (see also \cite{behboodian2005measures, fuchs2016copula, fuchs2014bivariate}) 
links the measure of concordance $\kappa_A$ to the invariant copula $A$.

\begin{proposition}[Characterization of copula-induced measures of concordance] \label{ConstructionConcordanceMeasure}~~\\
The map $\kappa_A$ is a measure of concordance if and only if $A$ is invariant.
In either case, $\kappa_A$ is convex.
In particular, 
if $A$ is invariant, then $[\Pi, A] = 1/4$ and the identity 
\begin{align}\label{Def.KappaA}
    \kappa_A(C) = \frac{[C,A] - 1/4}{[M,A] - 1/4}
\end{align}
holds for all $C \in \C$.
\end{proposition}

The class of measures of concordance given by \eqref{Def.KappaA} comprises the popular indices Spearman's rho and Gini's gamma and can also be used to construct new indices:

\begin{exa}[Well-known and new measures of concordance]~~\label{Ex.MOC.1}
    \begin{enumerate}
        \item \textbf{Spearman's rho:} 
        The copula $\Pi$ is invariant and $\kappa_\Pi$ satisfies
        \begin{align*}
            \kappa_\Pi(C) = 12 \, [C, \Pi] - 3
        \end{align*}
        which means that $\kappa_\Pi$ is Spearman's rho $\rho$; see \cite{fuchs2014bivariate, nelsen2007introduction}.
        
        \item \textbf{Gini's gamma:} The copula $M_{\Gamma} = (M+W)/2$ is invariant and $\kappa_{M_{\Gamma}}$ satisfies
        \begin{align*}
            \kappa_{M_{\Gamma}}(C) = 8 \, [C, M_{\Gamma}] - 2
        \end{align*}
        which means that $\kappa_{M_{\Gamma}}$ is Gini's gamma $\gamma$; see \cite{fuchs2014bivariate, nelsen2007introduction}.
        
        \item The copula $V$ is invariant and $\kappa_V$ fulfills
        \begin{align*}
            \kappa_V (C) = 16 \, [C, V] - 4
        \end{align*}
        
        \item \textbf{Linear interpolation:} For $\alpha \in [0,1]$, consider the Fr{\'e}chet copula
        \begin{align*}
            A_{\alpha} 
            := \frac{\alpha}{2} \, M + (1-\alpha) \, \Pi + \frac{\alpha}{2} \, W
            = \alpha \, M_\Gamma + (1-\alpha) \, \Pi
        \end{align*}
        Then $A_{\alpha}$ is an invariant copula and the induced measure of concordance $\kappa_{A_{\alpha}}$ satisfies
        \begin{align*}
            \kappa_{A_{\alpha}} (C) 
            & = \frac{2(1-\alpha)}{2+\alpha} \kappa_\Pi(C) + \frac{3\alpha}{2+\alpha} \kappa_{M_{\Gamma}} (C)
        \end{align*}
        meaning that $\kappa_{A_{\alpha}}$ is a weighted mean of Spearman's rho and Gini's gamma, 
        and for $\alpha \in (0,1)$ the respective weights are distinct from $1-\alpha$ and $\alpha$.
    \end{enumerate}
\end{exa}
\noindent The last example can be extended to the case of arbitrary invariant copulas in the place of $\Pi$ and $M_{\Gamma}$. 
Note that Kendall's tau can not be constructed using invariant copulas.

\bigskip
In \cite{durante2019reflection}, the authors have introduced a general approach to generate invariant copulas via a transformation $\vartheta: \mathcal{C} \to \mathcal{C}$ given by 
\begin{align*}
    \vartheta(C)(u,v) 
    & := \frac{1}{4} \, C(2u, 2v) \, \mathds{1}_{[0,1/2]^2}(u,v) 
         + \frac{1}{4} \, (\nu_1(C))(2u-1, 2v) \, \mathds{1}_{(1/2,1] \times [0,1/2]}(u,v) 
    \\
    & \phantom{00} + \frac{1}{4} \, (\nu_2(C))(2u, 2v-1) \, \mathds{1}_{[0,1/2] \times (1/2,1]}(u,v) 
         + \frac{1}{4} \, (\nu(C))(2u-1, 2v-1) \, \mathds{1}_{(1/2,1]^2 }(u,v) \,.
\end{align*}
According to \cite[Theorem 3.4]{durante2019reflection},
$C$ is symmetric if and only if $\vartheta(C)$ is invariant.
Figure \ref{ExampleVartheta} illustrates the idea behind the transformation $\vartheta$: 
first, the unit square $[0,1]^2$ is partitioned into the four subsets
$[0,1/2]^2 \cup [0,1/2] \times [1/2, 1] \cup [1/2, 1] \times [0,1/2] \cup [1/2, 1]^2$;
then, the symmetric copula $C$ is inserted into $[0,1/2]^2$ and reflected at $u=1/2$ and $v=1/2$ resulting in the invariant copula $\vartheta(C)$ (see again \cite{durante2019reflection}). 
For example, $\vartheta(M) = M_\Gamma$ as depicted in Figure \ref{ExampleVartheta}, $\vartheta(\Pi) = \Pi$ and $\vartheta(W) = V$.

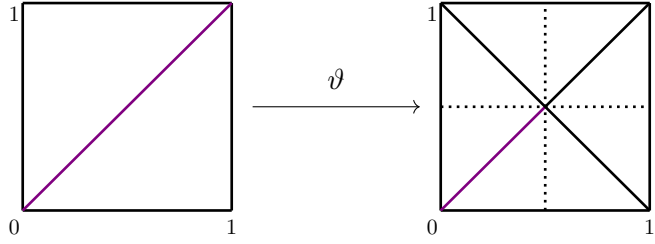
\begin{figure}[!ht]
    \centering
    \begin{tikzpicture}[scale=2.75]
	        \draw[-,line width=1] (0,0) -- (1,0);
	        \draw[-,line width=1] (0,0) -- (0,1);
	        \draw[-,line width=1] (1,0) -- (1,1);
	        \draw[-,line width=1] (0,1) -- (1,1);
	        \node[below=1pt of {(-0.04,0.02)}, scale= 0.75, outer sep=2pt] {$0$};
	        \node[below=1pt of {(1,0.02)}, scale= 0.75, outer sep=2pt] {$1$};
	        \node[below=1pt of {(-0.04,1.05)}, scale= 0.75, outer sep=2pt] {$1$};
	        \draw[-,line width=1, violet] (0,0) -- (1,1);
                \draw[->] (1.1,1/2) -- (1.9, 1/2) ;
                \node[above=1pt of {(1.5,1/2)}, scale= 1, outer sep=2pt,fill=white] {$\vartheta$};
	        \draw[-,line width=1] (2,0) -- (3,0);
	        \draw[-,line width=1] (2,0) -- (2,1);
	        \draw[-,line width=1] (3,0) -- (3,1);
	        \draw[-,line width=1] (2,1) -- (3,1);
	        \node[below=1pt of {(1.96,0.02)}, scale= 0.75, outer sep=2pt] {$0$};
	        \node[below=1pt of {(3,0.02)}, scale= 0.75, outer sep=2pt] {$1$};
	        \node[below=1pt of {(1.96,1.05)}, scale= 0.75, outer sep=2pt] {$1$};
	        \draw[-,line width=1, violet] (2,0) -- (2.5,0.5);
                \draw[-,line width=1] (2.5,0.5) -- (3,1);
                \draw[-,line width=1] (2,1) -- (3,0);
                \draw[-,line width=1, dotted] (2.5,0) -- (2.5,1);
                \draw[-,line width=1, dotted] (2,1/2) -- (3,1/2);
	    \end{tikzpicture}
    \caption{The copula $M$ and its transformed analogue $\vartheta(M) = M_\Gamma$.}
    \label{ExampleVartheta}
\end{figure}

Remarkably, the map $\vartheta$ preserves the concordance order in some sense  
(see \cite[Theorem 3.5, Theorem 3.9, Corollary 3.10]{durante2019reflection}):
\begin{proposition}[Order preserving properties of $\vartheta$]\label{LemmaVarthetaProperties}~~
    \begin{enumerate}
    \item 
    For every $C,D \in\C$ with $C \leq D$, 
    the inequality $ \vartheta(C)(u,v) \leq \vartheta(D)(u,v)$ holds for all $(u,v) \in [0,1/2]^2$.
    \item 
    For every invariant copula $A$ the inequality 
    $\vartheta(W)(u,v) = V(u,v) \leq A(u,v) \leq M_\Gamma(u,v) = \vartheta(M)(u,v)$
    holds for all $(u,v) \in [0,1/2]^2$.
    \item $\vartheta$ is a bijection between the class of symmetric copulas and the class of invariant copulas.
\end{enumerate}
\end{proposition}

Motivated by the fact that $\vartheta$ is bijective and order preserving according to Proposition \ref{LemmaVarthetaProperties}, 
we now define an ordering on invariant copulas.

\begin{definition}[An ordering for invariant copulas]~~\\
For two invariant copulas $A$ and $B$ we write 
$A \preceq B$ if $\vartheta^{-1}(A) \leq \vartheta^{-1}(B)$.
\end{definition}
\noindent 
Notice that $A \preceq B$ if and only if $A|_{[0,1/2]^2} \leq B|_{[0,1/2]^2}$.
$\preceq$ is reflexive, transitive and antisymmetric, hence an order relation on invariant copulas.
As a consequence of Proposition \ref{LemmaVarthetaProperties}, every invariant copula $A$ fulfills $V \preceq A \preceq M_\Gamma$.

\bigskip
So far we have examined how invariant copulas $A$ can be used to generate measures of concordance of the form 
\begin{align*}
    \kappa_A(C) = \alpha(A) \, \big( [C,A]-1/4 \big)
\end{align*}
with $\alpha(A) = \big([M,A]-1/4\big)^{-1} > 0$.
Now, we verify that the ordering $\preceq$ relates to an ordering of the induced measures of concordance.

\begin{theorem}[Comparison result for measures of concordance] \label{MainResultPMI.Thm}~~\\
For invariant copulas $A$ and $B$ with $A \preceq B$, the inequality
\begin{align} \label{Thm.PMI.Ineq}
  \alpha(A) \, \kappa_B(C) \geq \alpha(B) \, \kappa_A(C)
\end{align}
holds for all those $C\in\C$ that are PMI. 
\end{theorem}
\noindent 
Since $V \preceq \Pi \preceq M_\Gamma$,
from Theorem \ref{MainResultPMI.Thm} we immediately obtain the following comparison result for the measures of concordance Spearman's rho, Gini's gamma and $\kappa_V$.

\begin{corollary} \label{MainResultPMI.Cor}
The inequality
$$6\, \gamma(C) \geq 4\, \rho(C) \geq 3\, \kappa_V(C)$$
holds for all those $C\in\C$ that are PMI. 
\end{corollary}

\begin{remark} \label{MainResultPMI.Rem}
It is worth mentioning that Theorem \ref{MainResultPMI.Thm} has an NMI analogue: 
For invariant copulas $A$ and $B$ with $A \preceq B$, the inequality
$\alpha(A) \, \kappa_B(C) \leq \alpha(B) \, \kappa_A(C)$
holds for all those $C\in\C$ that are NMI. 
Consequently, the inequality
$6\, \gamma(C) \leq 4\, \rho(C) \leq 3\, \kappa_V(C)$
holds for all those $C\in\C$ that are NMI. 
\end{remark}

\section{Estimation}\label{SectionEstimation}

Estimators are proposed for the measures of concordance $\kappa_A$ studied in previous Section \ref{SectionConcordanceMeasureByCopula}.
Two different approaches for estimating these quantities are employed, 
one is the construction via \emph{empirical copulas}, which are, however, not copulas,
and the other is the intrinsic construction via \emph{empirical checkerboard copulas} (also known as \emph{empirical bilinear copulas}), see, e.g., \cite{genest2017asymptotic}.
\\
To this end, we consider a bivariate random vector $\textbf{X} = (X_1, X_2)$ with distribution function $F$, continuous univariate marginal distribution functions $F_1$ of $X_1$, $F_2$ of $X_2$ and connecting copula $C$.
Further, let $\textbf{X}_1, \dots, \textbf{X}_n$ be a random sample of i.i.d. copies from $\textbf{X} = (X_1, X_2)$,
and denote by $R_{i1}$ the rank of $X_{i1}$ among $X_{11}, \dots, X_{n1}$
and by $R_{i2}$ the rank of $X_{i2}$ among $X_{12}, \dots, X_{n2}$.
Since the univariate marginals are continuous ties only occur with probability $0$.

\subsection{Copula estimation}
Denote by $F_n: \R^2 \to \{0, \frac{1}{n}, \dots \frac{n-1}{n}, 1\}$ and $F_{n,1}, F_{n,2}: \R \to \{0, \frac{1}{n}, \dots \frac{n-1}{n}, 1\}$ the empirical distribution functions corresponding to $F$, $F_1$, and $F_2$, respectively, i.e., for $j \in \{1,2\}$,
\begin{align*}
    F_{n,j} (x_j) 
    = \frac{1}{n} \sum_{i=1}^n \mathds{1}_{(-\infty, x_j]} (X_{ij}) 
    \qquad \qquad &
    F_{n} (x_1,x_2) 
    = \frac{1}{n} \sum_{i=1}^n \prod_{j=1}^2 \mathds{1}_{(-\infty, x_j]} (X_{ij})
\end{align*}
for all $(x_1,x_2) \in \R^2$, so that $F_{n,j}(X_{ij}) = R_{ij}/n$ for all $i \in \{1,\dots,n\}$.
The \emph{empirical copula} then is defined, for all $(u_1, u_2) \in \I^2$, 
by
\begin{align*}
  C_n(u_1, u_2) 
  & = \frac{1}{n} \sum_{i=1}^n \prod_{j=1}^2 \mathds{1}_{[0, u_j]} \left( \frac{R_{ij}}{n+1} \right)
\end{align*}
and is (since we assume continuity of the univariate marginal distribution functions) asymptotically equivalent (see Remark \ref{Emp.Cop.Remark}) to the \emph{empirical checkerboard copula} defined, for all $(u_1, u_2) \in \I^2$, 
by (according to \cite{genest2017asymptotic})
\begin{align}\label{EmpBilinearCopula}
    \hat{C}_n(u_1, u_2) 
    := \frac{1}{n} \sum_{i=1}^n \prod_{j=1}^2 
       \big[ \beta_{F_{nj}}(u_j) \, \mathds{1}_{(-\infty,F_{n,j}^{-1}(u_j)]} (X_{ij}) 
             + (1-\beta_{F_{nj}}(u_j)) \mathds{1}_{(-\infty,F_{n,j}^{-1}(u_j))} (X_{ij}) \big]
\end{align}
where 
$G^{-1}$ denotes the pseudo-inverse of the distribution function $G$,
$G(.-)$ its left-hand limit and 
$$  
  \beta_{G}(u) := 
  \begin{cases}
     \frac{u - G(G^{-1}(u)-)}{G(G^{-1}(u)) - G(G^{-1}(u)-)} & \textrm{if } G(G^{-1}(u)) > G(G^{-1}(u)-)
     \\
     1 & \textrm{otherwise.}
  \end{cases} 
$$
Recall that empirical checkerboard copulas are absolutely continuous copulas with density $\hat{c}_n$ that is piecewise constant on the interior of each rectangle 
$\big[\frac{k-1}{n}, \frac{k}{n}\big] \times \big[\frac{l-1}{n}, \frac{l}{n}\big]$,
$k,l \in \{1,\dots,n\}$.

\begin{remark}\label{Emp.Cop.Remark}
It is worth mentioning that the term \enquote{empirical copula} is not uniquely defined and there exist different versions, all of which behave asymptotically the same.
Besides $C_n$, the literature also recognizes the empirical copulas
\begin{align}\label{Emp.Cop.Alternative}
    C_n^\ast (u_1, u_2) 
    & := F_n \big( F_{n,1}^{-1}(u_1), F_{n,2}^{-1}(u_2) \big)
    \\
    C_n^{\ast\ast} (u_1, u_2) 
    & := \frac{1}{n} \sum_{i=1}^n \prod_{j=1}^2 \mathds{1}_{[0, u_j]} (F_{n,j}(X_{ij}))
       = \frac{1}{n} \sum_{i=1}^n \prod_{j=1}^2 \mathds{1}_{[0, u_j]} \left(\frac{R_{ij}}{n}\right)
    \label{Emp.Cop.Classical}
\end{align}
(see, e.g, \cite{genest2017asymptotic, fermanian2004weak, janssen2012large, segers2012asymptotics})
where the latter is usually referred to as the \emph{classical empirical copula} and differs from $C_n$ only by a slightly varying normalization of the ranks.
All three versions are asymptotically equivalent, 
i.e., 
$$ 
  d_\infty(C_n,C_n^\ast) \leq \frac{2}{n} 
  \qquad \textrm{ and } \qquad
  d_\infty(C_n,C_n^{\ast\ast}) \leq \frac{2}{n} 
$$ 
(compare, e.g., \cite{fermanian2004weak})
and are asymptotically equivalent to the empirical checkerboard estimator $\hat{C}_n$
(see, e.g., \cite[Remark 2]{genest2017asymptotic}).
\end{remark}


In what follows, we briefly summarize the conditions under which the empirical processes
(see, e.g., \cite{genest2017asymptotic, segers2012asymptotics})
\begin{equation*}
  \mathbb{C}_n := \sqrt{n} \, (C_n - C)
  \qquad \textrm{and} \qquad
  \hat{\mathbb{C}}_n := \sqrt{n} \, (\hat{C_n} - C)
\end{equation*}
converge weakly to a centered Gaussian process.
The next condition which originates from \cite[Condition 2.1]{segers2012asymptotics} 
and involves the copula's first order partial derivatives $\dot{C}_1$ and $\dot{C}_2$,
is needed to ensure the convergence of the empirical copula (EC for short) process $\mathbb{C}_n$ and the empirical checkerboard copula (ECC for short) process $\hat{\mathbb{C}}_n$.

\begin{condition}\label{ConditionClassical} 
  For $j \in \{1,2\}$, the partial derivative $\dot{C}_j$ exists and is continuous on the set $\{ (u_1, u_2) \in \I^2: u_j \in (0,1)\}$.
\end{condition}
\noindent We note in passing that, in case the partial derivative $\dot{C}_j$, $j \in \{1,2\}$, exists, it coincides with the copula's Markov kernel almost surely (see, e.g., \cite{durante2016principles}).

Now, with Condition \ref{ConditionClassical} at hand, we are able to restate the key result concerning the convergence of the above introduced copula processes
to the limit
\begin{align}\label{GaussianLimit}
    \mathbb{C} (u_1, u_2) 
    & :=
    \mathbb{B}_{C}(u_1,u_2) - \mathbb{B}_{C}(u_1,1) \, \dot{C}_1(u_1,u_2) - \mathbb{B}_{C}(1,u_2) \, \dot{C}_2 (u_1,u_2)
\end{align}
where $\mathbb{B}_{C}$ denotes a centered Gaussian process on $\I^2$ such that
$\text{cov}(\mathbb{B}_{C}(\textbf{u}), \mathbb{B}_{C}(\textbf{v})) 
= C (\textbf{u} \wedge \textbf{v}) - C(\textbf{u}) \, C(\textbf{v})$
for all $\textbf{u} = (u_1, u_2), \textbf{v} = (v_1, v_2) \in [0,1]^2$;
here $\wedge$ denotes the component-wise minimum.
The next result is due to \cite[Proposition 3.1]{segers2012asymptotics} 
(see also \cite[Theorem 1 \& Corollary 1]{genest2017asymptotic})
and the asymptotic equivalence of the various copula estimators mentioned in Remark \ref{Emp.Cop.Remark}; 
here $l^\infty(\I^2)$ denotes the space of all bounded functions from $\I^2$ to $\mathbb{R}$ equipped with the topology of uniform convergence. 

\begin{proposition}[Convergence of copula processes]~~\\
Suppose that copula $C$ satisfies Condition \ref{ConditionClassical}.  
Then $\mathbb{C}_n$ and $\hat{\mathbb{C}}_n$ converge weakly to $\mathbb{C}$ in $l^\infty(\I^2)$. 
\end{proposition}





Due to Remark \ref{Emp.Cop.Remark},
the following strong consistency rates for classical empirical copulas and empirical checkerboard copulas can be derived from \cite{janssen2012large}.

\begin{proposition}[Consistency rates]\label{ConvergenceEmpCopula}~~\\
With probability $1$ we have
\begin{equation*}
    d_\infty(C_n,C) = \mathcal{O} \left( \sqrt{\frac{\log \log n}{n}} \right)
    \qquad \textrm{and} \qquad 
    d_\infty(\ec,C) = \mathcal{O} \left( \sqrt{\frac{\log \log n}{n}} \right)
\end{equation*}
\end{proposition}

\subsection{Estimation of measures of concordance}

We now focus on estimating measures of concordance $\kappa_A$ induced by invariant copulas $A$,
using the empirical copula $C_n$ and the empirical checkerboard copula $\hat{C}_n$ as plug-ins to \eqref{Def.KappaA}.
To this end, we extend the biconvex form defined in \eqref{Def.BiconvexForm} to a map
$[.,.]: l^\infty(\I^2) \times \mathcal{C} \to \R$ and observe the following results:

\begin{lemma}\label{EstimaorMoCClassicalAndCheckerboard}
For an invariant copula $A$ and a random sample $\textbf{X}_1, \dots, \textbf{X}_n$ of i.i.d. copies from $\textbf{X} = (X_1, X_2)$, the empirical copula and the empirical checkerboard copula fulfill
\begin{align*}
  [C_n,A]
  & = \frac{1}{n} \sum_{i=1}^n A \left( \tfrac{R_{i1}}{n+1}, \tfrac{R_{i2}}{n+1} \right)
  \\
  [\hat{C}_n,A] 
  & = \frac{1}{n} \sum_{i=1}^n n^2
      \int_{\big(\frac{R_{i1}-1}{n},\frac{R_{i1}}{n}\big) \times \big(\frac{R_{i2}-1}{n},\frac{R_{i2}}{n}\big)} A(u_1,u_2) \, \mathrm{d}  \lambda(u_1,u_2)
\end{align*}
and 
\begin{align*}
  [M_n,A]
  & = \frac{1}{n} \sum_{i=1}^n A \left(\frac{i}{n+1}, \frac{i}{n+1} \right) 
  \\
  [\hat{M}_n,A] 
  & =  \frac{1}{n} \sum_{i=1}^n n^2 \int_{\big(\frac{i-1}{n}, \frac{i}{n}\big)^2} A(u_1,u_2) \, \mathrm{d}  \lambda(u_1,u_2)
\end{align*}
In particular, 
$[M_n,A] > 1/4$ for every $n \geq 4$ and $[\hat{M}_n,A] > 1/4$ for every $n \geq 2$.
\end{lemma}

\noindent
Lemma \ref{EstimaorMoCClassicalAndCheckerboard} provides the key ingredients for defining suitable estimators for $\kappa_A(C)$ with $A$ being invariant based on a random sample $\textbf{X}_1, \dots, \textbf{X}_n$ of i.i.d. copies from $\textbf{X} = (X_1, X_2)$ with connecting copula $C$.
For $n \geq 4$, define 
\begin{align} \label{Def.KM.CEC}
  \kappa_{A,n} 
  & := \frac{[C_n,A] - \frac{1}{4}}{[M_n,A] - \frac{1}{4}}
     = \alpha_n(A) \; 
     \left(  \frac{1}{n} \sum_{i=1}^n A \left( \tfrac{R_{i1}}{n+1}, \tfrac{R_{i2}}{n+1} \right) - \frac{1}{4} \right) 
\end{align}
with
$ \alpha_n(A)^{-1} 
  := [M_n,A] - \frac{1}{4} 
   = \frac{1}{n} \sum_{i=1}^n A \left(\frac{i}{n+1}, \frac{i}{n+1} \right) - \frac{1}{4} $
and, for $n \geq 2$, define
\begin{align} \label{Def.KM.ECC}
  \hat{\kappa}_{A,n} 
  & := \frac{[\hat{C}_n,A] - \frac{1}{4}}{[\hat{M}_n,A] - \frac{1}{4}}
     = \hat{\alpha}_n(A) \; 
     \left( \frac{1}{n} \, \sum_{i=1}^n n^2
      \int_{\big(\frac{R_{i1}-1}{n},\frac{R_{i1}}{n}\big) \times \big(\frac{R_{i2}-1}{n},\frac{R_{i2}}{n}\big)} A(u_1,u_2) \, \mathrm{d}  \lambda(u_1,u_2) - \frac{1}{4} \right) 
\end{align}
with 
$\hat{\alpha}_n(A)^{-1} 
     := [\hat{M}_n,A] - \frac{1}{4}
      = \frac{1}{n} \, \sum_{i=1}^n n^2 \int_{\big(\frac{i-1}{n}, \frac{i}{n}\big)^2} A(u_1,u_2) \, \mathrm{d}  \lambda(u_1,u_2) - \frac{1}{4}$.

\bigskip
In Theorem \ref{Asymp.Norm.kappa} below we show asymptotic normality of $\kappa_{A,n}$ and $\hat{\kappa}_{A,n}$.
Before that, we analyse the behaviour of $\alpha_n(A)$ and $\hat{\alpha}_n(A)$ when $n$ tends to $\infty$.

\begin{lemma}\label{AlphaConvergenceSpeed}
We have $\vert \alpha_n(A) - \alpha(A) \vert 
  = \mathcal{O} \left( \frac{1}{n} \right) = \vert \hat{\alpha}_n(A) - \alpha(A) \vert$. 
\end{lemma}


\begin{theorem}[Asymptotic normality]~~\label{Asymp.Norm.kappa}\\
Suppose that copula $A$ is invariant and copula $C$ satisfies Condition \ref{ConditionClassical}.
Then  $\sqrt{n} \, (\kappa_{A,n} - \kappa_A(C))$ and $\sqrt{n} \, (\hat{\kappa}_{A,n} - \kappa_A(C))$ converge weakly to the centered Gaussian random variable
\begin{align*}
  \alpha(A) \, [\mathbb{C},A] 
  = \alpha(A) \int_{\I^2} \mathbb{C}(\textbf{u}) \, \mathrm{d} \mu_A(\textbf{u})
\end{align*}
with variance 
\begin{align}\label{Asymp.Norm.kappa.Variance}
  \sigma_A^2 
  & := \alpha(A)^2 \int_{\I^2} \int_{\I^2} {\rm cov} (\mathbb{C}(\textbf{u}), \mathbb{C}(\textbf{v})) \, \, \mathrm{d} \mu_A(\textbf{u}) \, \mathrm{d} \mu_A(\textbf{v})\,. 
\end{align}
\end{theorem}

\begin{remark} \leavevmode
\begin{enumerate}
\item 
It is worth mentioning that, by following the approach developed in \cite{genest2017asymptotic},
the asymptotic normality established in Theorem \ref{Asymp.Norm.kappa} can be extended to copulas for which there exists an open set $\Lambda \subseteq (0,1)^2$ such that $\mu_A(\Lambda) = 1$ and on which $\dot{C}_1$ and $\dot{C}_2$ exist and are continuous.  

\item 
As an immediate consequence of Proposition \ref{ConvergenceEmpCopula} and Lemma \ref{AlphaConvergenceSpeed}, 
the two estimators $\kappa_{A,n}$ and $\hat{\kappa}_{A,n}$ are strongly consistent.
\end{enumerate}
\end{remark}

Example \ref{Ex.MOC.2} provides concrete and simple representations for (most of) the estimators $\kappa_{A,n}$ and $\hat{\kappa}_{A,n}$ when $A \in \{\Pi, M_\Gamma, V\}$;
unfortunately, we were not able to produce a handsome representation for $\hat{\kappa}_{V,n}$.

\begin{exa}[Estimators for measures of concordance]~~\label{Ex.MOC.2} \\
According to Example \ref{Ex.MOC.1}, 
$\kappa_\Pi$ equals Spearman's rho $\rho$ and $\kappa_{M_\Gamma}$ equals Gini's gamma $\gamma$.
Straightforward but rather tedious calculations lead to the following simple representations for the estimators of 
$\kappa_\Pi$, $\kappa_{M_\Gamma}$ and $\kappa_V$.
\begin{enumerate}
\item 
For \textbf{Spearman's rho}, we have 
\begin{equation*}
    \alpha_n(\Pi) = 12 \, \frac{(n+1)^2}{n^2-1}
    \qquad \qquad
    \hat{\alpha}_n(\Pi) = 12 \, \frac{n^2}{n^2-1}
\end{equation*} 
and
\begin{align*}
  \kappa_{\Pi,n} 
  & = \hat{\kappa}_{\Pi,n} 
    = 1 - \frac{6}{n(n^2-1)} \sum_{i=1}^n (R_{i1} - R_{i2})^2
    = \frac{12}{n(n^2-1)} \, \sum_{i=1}^n \, R_{i1}\, R_{i2} - 3 \, \frac{n+1}{n-1} 
\end{align*}
which shows that the estimators are just the well-known sample version of Spearman’s rho; 
compare \cite{kruskal1958, joe1990, perez2016, schmid2006multivariate}.
Interestingly, both estimators of Spearman's rho coincide.

\item 
For \textbf{Gini's gamma}, we have
\begin{equation*}
  \alpha_n(M_\Gamma) 
  = \frac{4n(n+1)}{\left\lfloor n^2/2 \right\rfloor}
  \qquad \qquad
  \hat{\alpha}_n(M_\Gamma)
  = \frac{6n^2}{\left\lfloor \frac{n(3n-2)}{4} \right\rfloor}
\end{equation*} 
and
\begin{align*}
  \kappa_{M_\Gamma,n} 
  & = \frac{1}{\lfloor n^2/2 \rfloor} \, 
      \left( \sum_{i=1}^n \left| R_{i1} + R_{i2} - (n+1) \right|
           - \sum_{i=1}^n \left| R_{i1} - R_{i2} \right| \right)
  \\
  \hat{\kappa}_{M_\Gamma,n}
  & = \frac{3}{2 \, \left\lfloor \frac{n(3n-2)}{4} \right\rfloor} \;  
            \left( \sum_{i=1}^n \left|R_{i1}+R_{i2} -(n+1)) \right|  
            - \sum_{i=1}^n \left|R_{i1}-R_{i2} \right|          
            + \frac{1}{3} \left( \sum_{\substack{1\leq i \leq n:\\ R_{i1} + R_{i2} = n + 1}} - \sum_{\substack{1\leq i\leq n:\\ R_{i1} = R_{i2}}} \right)   
            \right)
\end{align*}
which shows that the estimator $\kappa_{M_\Gamma,n}$ is just the sample version of Gini’s gamma; see, e.g., \cite[Section 5.1.4]{nelsen2007introduction} and \cite{genest2010spearman}).
Notice that the estimators differ, in general.

\item For $\kappa_V$, we have
\begin{equation*}
  \alpha_n(V) 
  = \frac{2n(n+1)}{\lfloor \frac{1}{8} (n-1)^2\rfloor}
\end{equation*}
and 
\begin{align*} 
  \kappa_{V,n} 
  & = \frac{n(n+1)}{2 \, \lfloor \frac{1}{8} (n-1)^2\rfloor} \, 
       - \frac{1}{\lfloor \frac{1}{8} (n-1)^2\rfloor} \sum_{i=1}^n
    \begin{cases}
      \left\vert R_{i1} - R_{i1} \right\vert
      & \text{ if } \left\vert R_{i1} - R_{i2} \right\vert > \frac{n+1}{2} 
      \\
      (n+1) - \left\vert R_{i1} + R_{i1} - (n+1) \right\vert 
      & \text{ if } \left\vert R_{i1} + R_{i2} - (n+1) \right\vert > \frac{n+1}{2} 
      \\
      \frac{n+1}{2} 
      & \text{ otherwise.}
    \end{cases} 
\end{align*}
We were not able to produce a handsome representation for the estimator $\hat{\kappa}_{V,n}$.
\end{enumerate}
\end{exa}

In Section \ref{SectionAsymptoticTesting} below we construct an asymptotic test for detecting whether Ineq. \eqref{Thm.PMI.Ineq} in Theorem \ref{MainResultPMI.Thm} holds.
This test can be used for detecting whether the underlying dependence structure of a given sample is PMI.
For this purpose, we need the following result which shows that the processes 
$$ \alpha_n(B) \sqrt{n} \, (\kappa_{A,n} - \kappa_A(C))  
  - \alpha_n(A) \sqrt{n} \, (\kappa_{B,n} - \kappa_B(C))$$
and 
$$ \hat{\alpha}_n(B) \sqrt{n} \, (\hat{\kappa}_{A,n} - \kappa_A(C)) 
  - \hat{\alpha}_n(A) \sqrt{n} \, (\hat{\kappa}_{B,n} - \kappa_B(C))$$
are asymptotically normal.
Notice that, since the estimators for $\kappa(A)$ and $\kappa(B)$ are not independent, in general, we cannot simply merge the two limit distributions for $\sqrt{n} \, (\kappa_{A,n} - \kappa_A(C))$ and $\sqrt{n} \, (\kappa_{B,n} - \kappa_B(C))$ derived in Theorem \ref{Asymp.Norm.kappa}.

\begin{theorem} \label{Test.Thm}
Suppose that copulas $A$ and $B$ are invariant 
and copula $C$ satisfies Condition \ref{ConditionClassical}. 
Then
$ \alpha(B) \sqrt{n} \, (\kappa_{A,n} - \kappa_A(C)) 
  - \alpha(A) \sqrt{n} \, (\kappa_{B,n} - \kappa_B(C))$
and 
$ \alpha(B) \sqrt{n} \, (\hat{\kappa}_{A,n} - \kappa_A(C)) 
  - \alpha(A) \sqrt{n} \, (\hat{\kappa}_{B,n} - \kappa_B(C))$
converge weakly to the centered Gaussian random variable
\begin{align}\label{Asymp.Norm.Limit.Process}
    \alpha(A) \alpha(B) \, \Big( [\mathbb{C},A] - [\mathbb{C},B] \Big)
\end{align}
with variance
\begin{align}\label{Asymp.Norm.Limit.Variance}
  \sigma_{A,B}^2 
  & := \alpha(A)^2 \alpha(B)^2 
       \left( \int_{\I^2} \int_{\I^2} {\rm cov} (\mathbb{C}(\textbf{u}), \mathbb{C}(\textbf{v})) \, \, \mathrm{d} \mu_A(\textbf{u}) \, \mathrm{d} \mu_A(\textbf{v}) +\int_{\I^2} \int_{\I^2} {\rm cov} (\mathbb{C}(\textbf{u}), \mathbb{C}(\textbf{v})) \, \, \mathrm{d} \mu_B(\textbf{u}) \, \mathrm{d} \mu_B(\textbf{v}) \right)
  \\
  & \qquad - 2 \, \alpha(A)^2 \alpha(B)^2 
      \int_{\I^2} \int_{\I^2} {\rm cov} (\mathbb{C}(\textbf{u}), \mathbb{C}(\textbf{v})) \, \, \mathrm{d} \mu_A(\textbf{u}) \, \mathrm{d} \mu_B(\textbf{v})\,. \notag
\end{align}
\end{theorem}

\section{Asymptotic testing \label{SectionAsymptoticTesting}}

Building upon the main result Theorem \ref{MainResultPMI.Thm} asymptotic tests are constructed for detecting whether Ineq. \eqref{Thm.PMI.Ineq} holds. 
These tests have the potential of being used for detecting whether the underlying dependence structure of a given sample is PMI (or NMI).
This is of particular interest in practice since, for example in the case of a rejection, certain families of PMI copulas such as Gaussian copulas, EFGM copulas, Frank copulas (or as we conject those Archimedean copulas that are PQD) may be excluded for model building.
\\
Consider again a bivariate random vector $\textbf{X} = (X_1, X_2)$ with continuous univariate marginal distribution functions $F_1$ of $X_1$, $F_2$ of $X_2$ and connecting copula $C$.
Further, let $\textbf{X}_1, \dots, \textbf{X}_n$ be a random sample of i.i.d. copies from $\textbf{X} = (X_1, X_2)$.

\bigskip
According to Theorem \ref{MainResultPMI.Thm}, 
in case the copula $C$ is PMI, the inequality $\alpha(A) \, \kappa_B(C) \geq \alpha(B) \, \kappa_A(C)$
holds for every two invariant copulas $A$ and $B$ with $A \preceq B$.
We therefore propose to use the test 
\begin{equation*}
    H_0: \alpha(A) \, \kappa_B(C) \geq \alpha(B) \, \kappa_A(C)
    \qquad \text{ vs. } \qquad
    H_1: \alpha(A) \, \kappa_B(C) < \alpha(B) \, \kappa_A(C)
\end{equation*}
for evaluating the property PMI of a given dependence structure.
As test statistic we then use either
\begin{equation} \label{TestStat.}
  T_{A,B,n} 
  := \sqrt{n} \;\frac{\alpha(B) \, \kappa_{A,n} - \alpha(A) \, \kappa_{B,n}}{\sigma_{A,B,n}}
  \qquad \text{ or } \qquad 
  \hat{T}_{A,B,n} 
  := \sqrt{n} \;\frac{\alpha(B) \, \hat{\kappa}_{A,n} - \alpha(A) \, \hat{\kappa}_{B,n}}{\hat{\sigma}_{A,B,n}}
\end{equation}
where $\sigma_{A,B,n}^2 > 0$ and $\hat{\sigma}_{A,B,n}^2 > 0$ denote consistent estimators for the unknown variance $\sigma_{A,B}^2$ of the limiting distribution $\alpha(A) \alpha(B) \, ([\mathbb{C}, A] - [\mathbb{C}, B])$ presented in \eqref{Asymp.Norm.Limit.Variance}.
Then, the rejection rule $T_{A,B,n} > z_{1-\alpha}$ respectively $\hat{T}_{A,B,n} > z_{1-\alpha}$, where $z_{1-\alpha}$ denotes the $1-\alpha$-quantile of the standard normal distribution, asymptotically rejects the null hypothesis $H_0$ on significance level at most $\alpha$. 
More precisely, under the null we have \pagebreak
\begin{align*}
  \PP_0 (T_{A,B,n} > z_{1-\alpha}) 
  &   =  \PP_0 \left( \sqrt{n} \;\frac{\alpha(B) \, \kappa_{A,n} - \alpha(A) \, \kappa_{B,n}}{\sigma_{A,B,n}}  > z_{1-\alpha} \right) 
  \\ 
  & \leq \PP_0 \left( \sqrt{n}\; \frac{\alpha(B) \, \kappa_{A,n} - \alpha(A) \, \kappa_{B,n} + \alpha(A) \, \kappa_B(C) - \alpha(B) \, \kappa_A(C)}{\sigma_{A,B,n}}  > z_{1-\alpha} \right) 
  \\
  &   =  \PP_0 \left(\frac{\alpha(B) \sqrt{n} \, (\kappa_{A,n} - \kappa_A(C)) 
  - \alpha(A) \sqrt{n} \, (\kappa_{B,n} - \kappa_B(C))}{\sigma_{A,B,n}}  > z_{1-\alpha} \right)
   \\
  & \overset{\text{d}}{\to} 1 - \Phi(z_{1-\alpha}) = \alpha\,,
\end{align*}
where the weak convergence is due to Theorem \ref{Test.Thm} and Slutsky's theorem; 
analogously for the test statistic $\hat{T}_{A,B,n}$.

\begin{remark}~~\label{Test.Rem}
\begin{enumerate}
\item 
Since the variance $\sigma_{A,B}^2$ \eqref{Asymp.Norm.Limit.Variance} of the limiting distribution \eqref{Asymp.Norm.Limit.Process} is typically unknown and depends on the partial derivatives of the unknown copula (see \eqref{GaussianLimit}), 
we here use a multiplier bootstrap approach for approximating $\sigma_{A,B}^2$
(see, e.g., \cite{dette2010,remillard2009,genest2017asymptotic}).
In a nutshell, following the multiplier bootstrap algorithm presented in \cite{dette2010} for the empirical copula and in \cite{genest2017asymptotic} for the empirical checkerboard copula, respectively, we construct a sample of the limiting process \eqref{GaussianLimit} which leads, applying the same arguments as in Theorem \ref{Test.Thm}, to a sample of the limiting distribution \eqref{Asymp.Norm.Limit.Process} whose empirical variance $\sigma_{A,B,n}^2$ and $\hat{\sigma}_{A,B,n}^2$, respectively, we use as estimate for $\sigma_{A,B}^2$.

\item
It is worth mentioning that, due to Remark \ref{MainResultPMI.Rem}, an asymptotic test for property NMI can be constructed using test statistics $- T_{A,B,n}$ and $- \hat{T}_{A,B,n}$.
\end{enumerate}
\end{remark}

\subsection{Simulation study}

We now illustrate the finite sample performance of the above introduced asymptotic tests by means of a simulation study in case the connecting copula of the random variables $X$ and $Y$ is either a
\begin{enumerate}
    \item Gaussian copula $C_\rho$ with varying parameter $\rho \in \{-0.95,-0.75,-0.5,-0.15,0,0.15,0.5,0.75,0.95\}$ where $C_0=\Pi$, or a
    \item Frank copula $C_\delta$ with varying parameter $\delta \in \{-20,-10,-5,-1,0,1,5,10,20\}$ where $C_0=\Pi$.
\end{enumerate}
Recall that, according to Examples \ref{Cor.Gaussian} and \ref{ExampleFrank},
$C_\rho$ is PMI if and only if $\rho \in (0,1)$ and $C_\delta$ is PMI if and only if $\delta \in (0,\infty)$.
Three different asymptotic tests are designed based on the comparison of the three measures of concordance
\begin{enumerate}
    \item[(T1)] Spearman's rho and Gini's gamma, 
    \item[(T2)] Spearman's rho and $\kappa_V$, and
    \item[(T3)] Gini's gamma and $\kappa_V$,
\end{enumerate}
where the estimators are either based on the empirical copula $C_n$ or the empirical checkerboard copula $\hat{C}_n$, and the significance level is set to $\alpha=0.05$.

Figures \ref{Fig.Test} and \ref{Fig.Test2} depict the rejection rates for the Gaussian and Frank family for different parameter and varying sample size $n \in \{50, 100, 250, 1000\}$. 
The results are similar for both dependence structures and both estimation principles considered.
We can conclude that the tests for evaluating Ineq. \eqref{Thm.PMI.Ineq} are slightly liberal but become more accurate with increasing sample size.
In view of testing the property PMI, the tests appear to be \enquote{overconservative}.
We also observe an increase in power with increasing sample size.
\begin{figure}[!ht]
  \centering
  \includegraphics[width=1\textwidth]{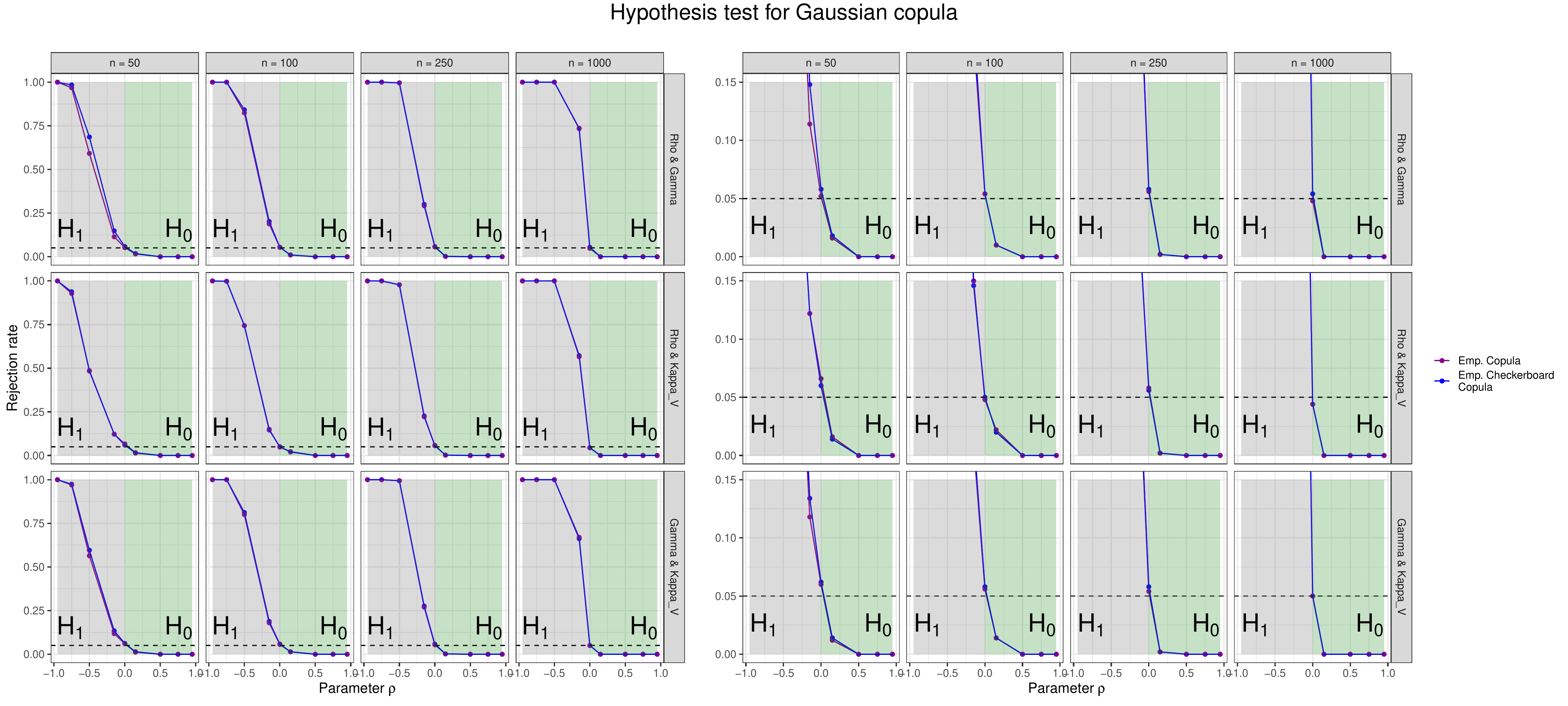}
  \caption{Gauss copula: Rejection rates of hypothesis tests ($\alpha = 0.05$) based on 
  (T1) Spearman's rho \& Gini's gamma, 
  (T2) Spearman's rho \& $\kappa_V$, and 
  (T3) Gini's gamma \& $\kappa_V$ for different parameter and varying sample size $n \in \{50, 100, 250, 1000\}$.} 
  \label{Fig.Test}
\end{figure}
\begin{figure}[!ht]
  \centering
  \includegraphics[width=1\textwidth]{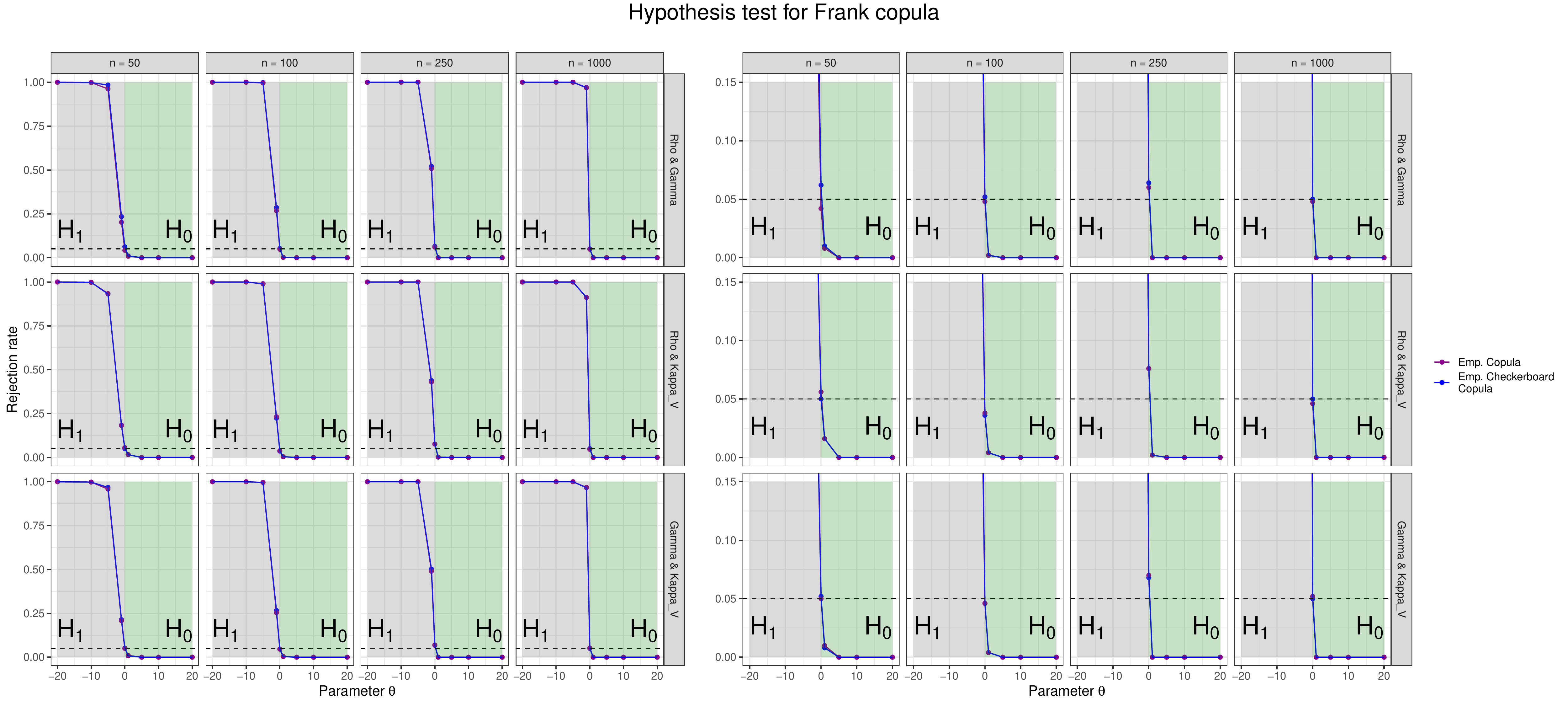}
  \caption{Frank copula: Rejection rates of hypothesis tests ($\alpha = 0.05$) based on 
  (T1) Spearman's rho \& Gini's gamma, 
  (T2) Spearman's rho \& $\kappa_V$, and 
  (T3) Gini's gamma \& $\kappa_V$ for different parameter and varying sample size $n \in \{50, 100, 250, 1000\}$.} 
  \label{Fig.Test2}
\end{figure}

\noindent
Within each test, the unknown variance is estimated via multiplier bootstrap where, for each sample, $R= 1,000$ bootstrap replicates are calculated.
As illustrated in Figure \ref{Fig.Var} for Gaussian and Frank copulas, 
the variance strongly depends on the concrete test design and the parameter used, but rather little on the underlying estimation principle.
\begin{figure}[!ht]
  \centering
  \includegraphics[width=1\textwidth]{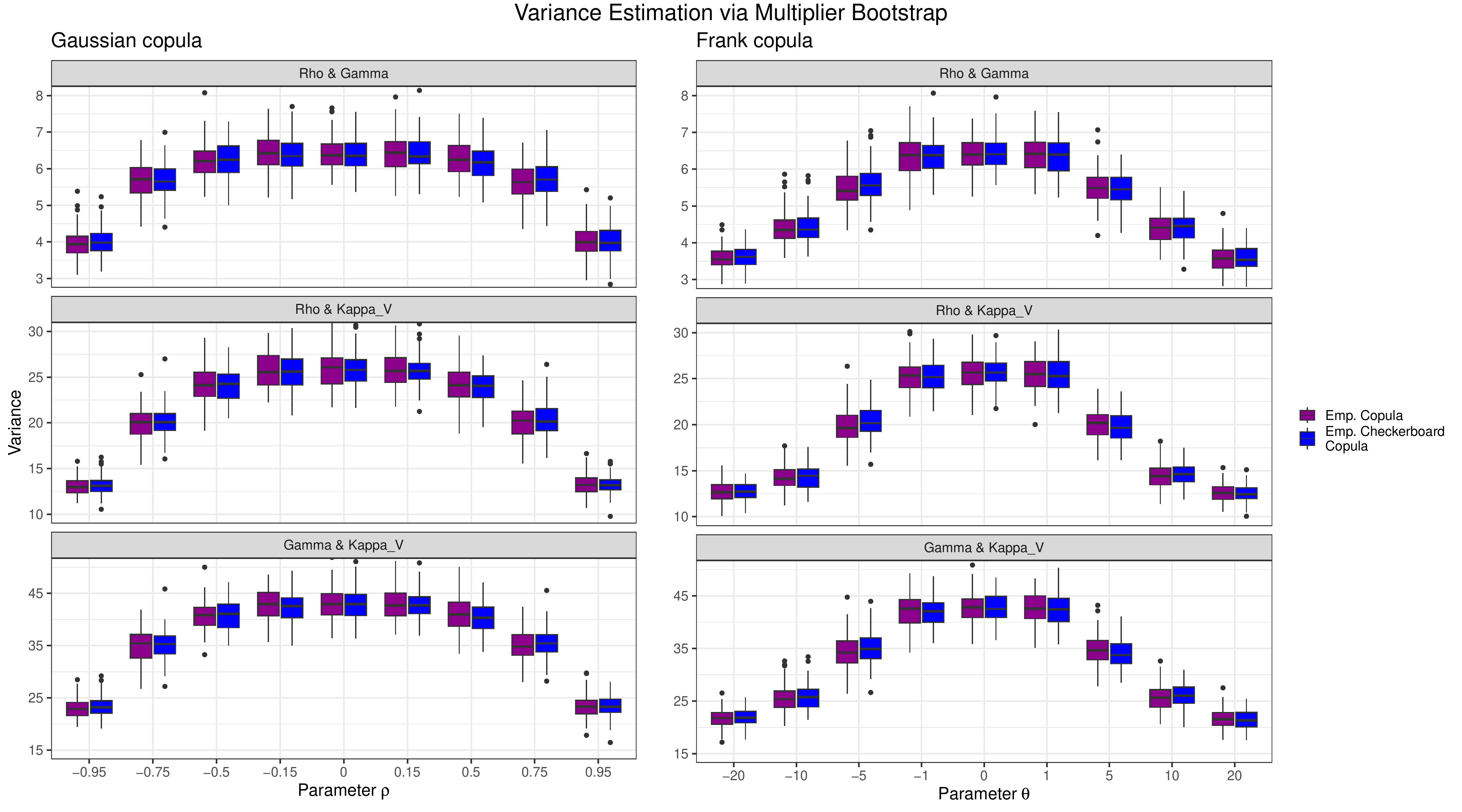}
  \caption{Boxplots summarizing the $100$ obtained variance estimates when using 
  (T1) Spearman's rho \& Gini's gamma, 
  (T2) Spearman's rho \& $\kappa_V$, and 
  (T3) Gini's gamma \& $\kappa_V$ for different parameter and sample size $n = 500$.} 
  \label{Fig.Var}
\end{figure}

\subsection{Real data example}
Finally, we illustrate the potential and importance of the introduced asymptotic tests by analyzing two real data sets.

First, let us consider the data set {\sf faithful} provided in the R package {\sf datasets}. 
The data set contains $n = 272$ observations of the waiting times between eruptions
(variable {\sf waiting}) and the duration of the eruption (variable {\sf eruptions}) for the Old Faithful geyser in Yellowstone National Park, Wyoming, USA.
Right panel of Figure \ref{Fig.Geyser} depicts the dependence structure between the variables {\sf waiting} and {\sf eruptions}, which resembles an ordinal sum structure of $\Pi$ and may therefore be regarded as PMI.
\begin{figure}[!ht]
  \centering
  \includegraphics[width=0.625\textwidth]{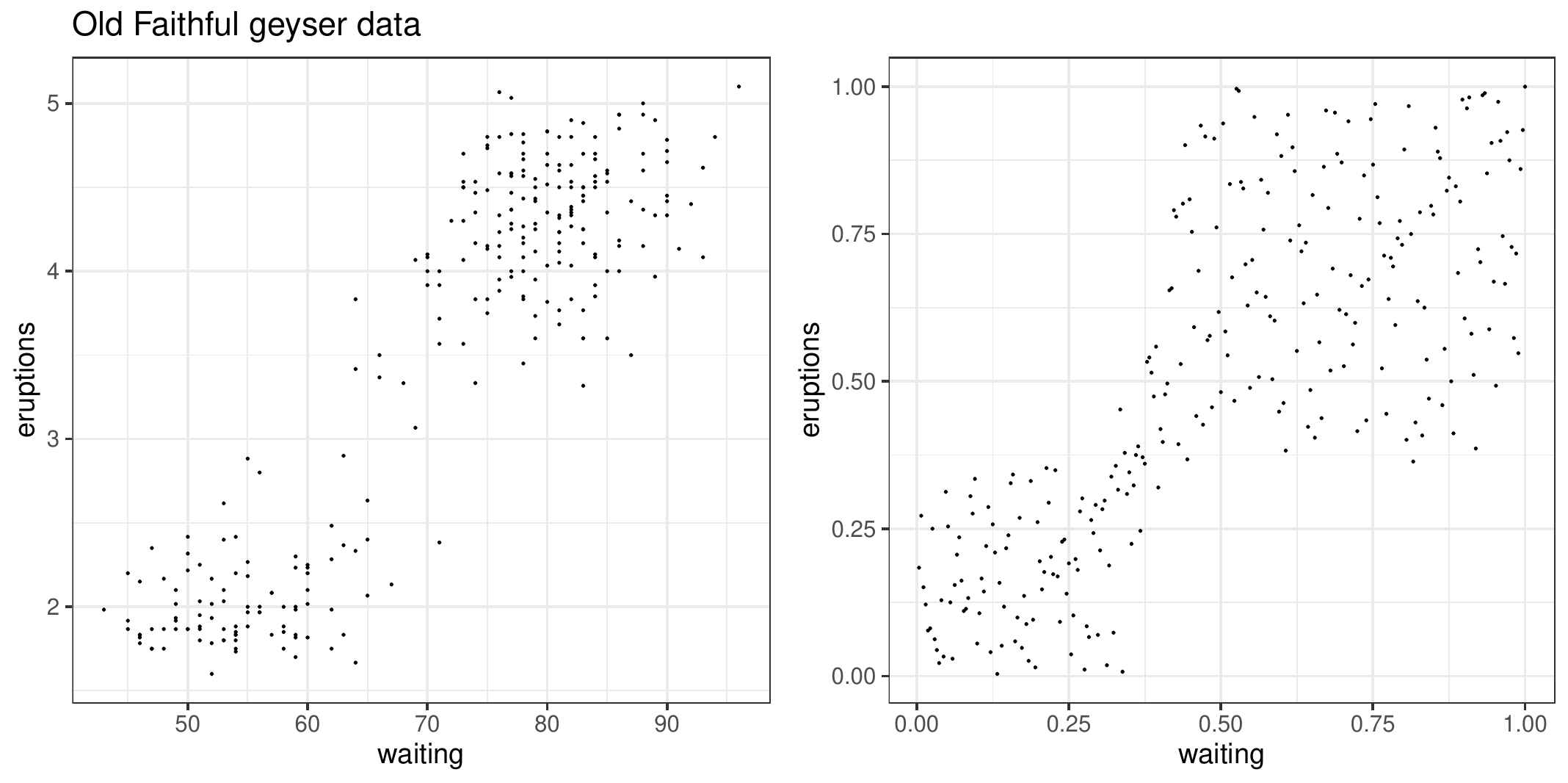}
  \caption{Observations (left panel) and pseudo-observations (right panel) of data set faithful.} 
  \label{Fig.Geyser}
\end{figure}

\noindent
All three tests (T1), (T2) and (T3) based on the empirical copula do not reject the null hypothesis
(test statistics (T1): $-14.35$, (T2): $-13.39$, (T3): $-14.59$; threshold: $1.64$), so that also from a statistical point of view the assumption PMI appears reasonable.
As mentioned in Remark \ref{Test.Rem}, the test statistic can also be used to test for the property NMI, in which case the null hypothesis is rejected regardless of the test chosen (test statistics (T1): $14.35$, (T2): $13.39$, (T3): $14.59$; threshold: $1.64$).

Second, we consider a data set of bioclimatic variables for $n = 1862$ locations homogeneously distributed over the global landmass from CHELSEA (\cite{karger2017,karger2018})
and restrict to the variables {\sf temperature seasonality} (TS) and {\sf precipitation seasonality} (PS).
Right panel of Figure \ref{Fig.Chelsea} depicts the dependence structure between the variables TS and PS.
Although their Spearman's rank correlation equals $0.0013$, the value of Gini's gamma is $-0.0152$ and their Kendall correlation equals $0.02314$,  
the variables are not independent.
All three tests (T1), (T2) and (T3) based on the empirical copula reject the null hypothesis
(test statistics (T1): $3.43$, (T2): $1.76$, (T3): $2.70$; threshold: $1.64$), 
so we may conclude that the two variables fail to be PMI and thus certain copula (sub)classes such as Gaussian copulas, FGM copulas, Fr{\'e}chet copulas, Frank copulas or even those Archimedean copulas that are PQD are not suitable for modelling the underlying dependence structure.
\begin{figure}[!ht]
  \centering
  \includegraphics[width=0.625\textwidth]{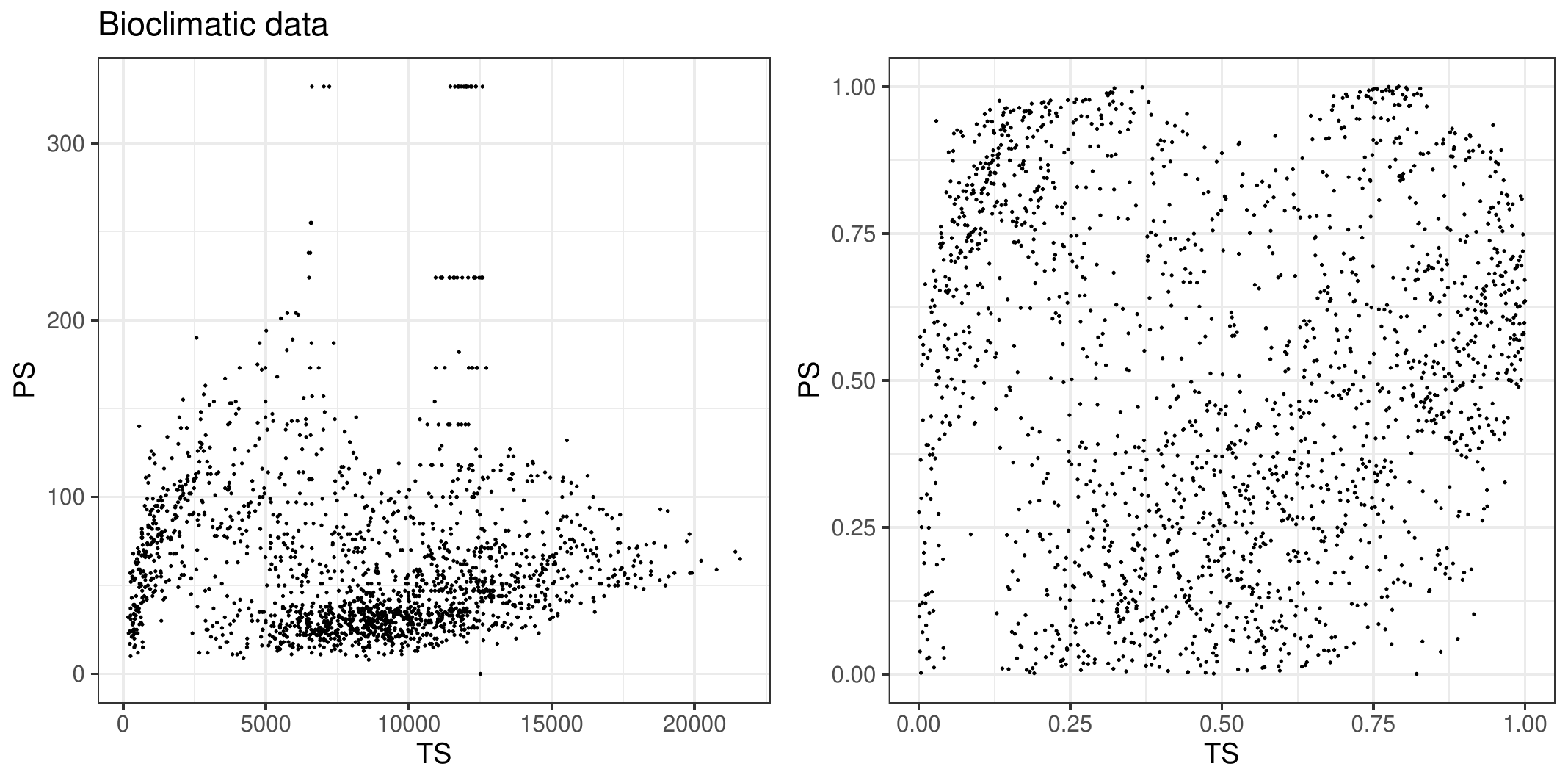}
  \caption{Observations (left panel) and pseudo-observations (right panel) of variables TS and PS of data set CHELSEA.} 
  \label{Fig.Chelsea}
\end{figure}

\section*{Acknowledgement}

SF wishes to thank Klaus D. Schmidt for helpful discussions.
SF further gratefully acknowledges the support of the Austrian Science Fund (FWF) project
{P 36155-N} \emph{ReDim: Quantifying Dependence via Dimension Reduction}
and the support of the WISS 2025 project 'IDA-lab Salzburg' (20204-WISS/225/197-2019 and 20102-F1901166-KZP).
MT gratefully acknowledges the financial support from AMAG Austria Metall AG within the project ProSa.


\begin{thebibliography}{}

\bibitem[\protect\citeauthoryear{Ansari}{Ansari}{2022}]{ansari2023}
Ansari, J. (2022).
\newblock On a version of a multivariate integration by parts formula for
  {L}ebesgue integrals.
\newblock Available at \url{https://arxiv.org/abs/2203.06772}.

\bibitem[\protect\citeauthoryear{Ansari and Fuchs}{Ansari and
  Fuchs}{2023}]{ansari2022}
Ansari, J. and S.~Fuchs (2023).
\newblock A simple extension of {A}zadkia \& {C}hatterjee's rank correlation to
  a vector of endogenous variables.
\newblock Available at \url{arxiv.org/abs/2212.01621}.

\bibitem[\protect\citeauthoryear{Av{\'e}rous, Genest, and Kochar}{Av{\'e}rous
  et~al.}{2005}]{kochar2005}
Av{\'e}rous, J., C.~Genest, and C.~Kochar (2005).
\newblock On the dependence structure of order statistics.
\newblock {\em J. Multivariate Anal.\/}~{\em 94}, 159--171.

\bibitem[\protect\citeauthoryear{Behboodian, Dolati, and
  {\'U}beda-Flores}{Behboodian et~al.}{2005}]{behboodian2005measures}
Behboodian, J., A.~Dolati, and M.~{\'U}beda-Flores (2005).
\newblock Measures of association based on average quadrant dependence.
\newblock {\em J. Probab. Statist. Sci\/}~{\em 3}, 161--173.

\bibitem[\protect\citeauthoryear{B{\"u}cher and Dette}{B{\"u}cher and
  Dette}{2010}]{dette2010}
B{\"u}cher, A. and H.~Dette (2010).
\newblock A note on bootstrap approximations for the empirical copula process.
\newblock {\em Stat. Prob. Lett.\/}~{\em 80}, 1925--1932.

\bibitem[\protect\citeauthoryear{Cap{\'e}ra{\`a} and Genest}{Cap{\'e}ra{\`a}
  and Genest}{1993}]{caperaa1993spearman}
Cap{\'e}ra{\`a}, P. and C.~Genest (1993).
\newblock Spearman's $\rho$ is larger than {K}endall's $\tau$ for positively
  dependent random variables.
\newblock {\em J. Nonparametr. Stat.\/}~{\em 2\/}(2), 183--194.

\bibitem[\protect\citeauthoryear{Chatterjee}{Chatterjee}{2020}]{chatterjee2020}
Chatterjee, S. (2020).
\newblock A new coefficient of correlation.
\newblock {\em J. Amer. Statist. Ass.\/}~{\em 116\/}(536), 2009--2022.

\bibitem[\protect\citeauthoryear{Chen}{Chen}{2007}]{chen2007note}
Chen, Y.-P. (2007).
\newblock A note on the relationship between {S}pearman's $\rho$ and
  {K}endall's $\tau$ for extreme order statistics.
\newblock {\em J. Statist. Plann. Inference\/}~{\em 137\/}(7), 2165--2171.

\bibitem[\protect\citeauthoryear{Daniels}{Daniels}{1950}]{daniels1950}
Daniels, H. (1950).
\newblock Rank correlation and population models.
\newblock {\em J. R. Stat. Soc. Ser. B. Stat. Methodol.\/}~{\em 12}, 171--181.

\bibitem[\protect\citeauthoryear{De~Haan and Resnick}{De~Haan and
  Resnick}{1977}]{de1977limit}
De~Haan, L. and S.~I. Resnick (1977).
\newblock Limit theory for multivariate sample extremes.
\newblock {\em Zeitschrift f{\"u}r Wahrscheinlichkeitstheorie und verwandte
  Gebiete\/}~{\em 40\/}(4), 317--337.

\bibitem[\protect\citeauthoryear{Dette, Siburg, and Stoimenov}{Dette
  et~al.}{2013}]{siburg2013}
Dette, H., K.~F. Siburg, and P.~A. Stoimenov (2013).
\newblock A copula-based non-parametric measure of regression dependence.
\newblock {\em Scand. J. Statist.\/}~{\em 40\/}(1), 21--41.

\bibitem[\protect\citeauthoryear{Durante and Fuchs}{Durante and
  Fuchs}{2019}]{durante2019reflection}
Durante, F. and S.~Fuchs (2019).
\newblock Reflection invariant copulas.
\newblock {\em Fuzzy Set. Syst.\/}~{\em 354}, 63--73.

\bibitem[\protect\citeauthoryear{Durante and Sempi}{Durante and
  Sempi}{2016}]{durante2016principles}
Durante, F. and C.~Sempi (2016).
\newblock {\em Principles of Copula Theory}.
\newblock CRC Press Boca Raton, FL.

\bibitem[\protect\citeauthoryear{Durbin and Stuart}{Durbin and
  Stuart}{1951}]{durbin1951inversions}
Durbin, J. and A.~Stuart (1951).
\newblock Inversions and rank correlation coefficients.
\newblock {\em J. R. Stat. Soc. Ser. B. Stat. Methodol.\/}~{\em 13\/}(2),
  303--309.

\bibitem[\protect\citeauthoryear{Edwards, Mikusi{\'n}ski, and Taylor}{Edwards
  et~al.}{2004}]{edwards2004measures}
Edwards, H.~H., P.~Mikusi{\'n}ski, and M.~D. Taylor (2004).
\newblock Measures of concordance determined by ${D}_{4}$--invariant copulas.
\newblock {\em Int. J. Math. Math. Sci.\/}~{\em 2004(70)}, 3867--3875.

\bibitem[\protect\citeauthoryear{Edwards, Mikusi{\'n}ski, and Taylor}{Edwards
  et~al.}{2005}]{edwards2005measures}
Edwards, H.~H., P.~Mikusi{\'n}ski, and M.~D. Taylor (2005).
\newblock Measures of concordance determined by {$D\sb 4$}-invariant measures
  on {$(0,1)\sp 2$}.
\newblock {\em Proc. Amer. Math. Soc.\/}~{\em 133\/}(5), 1505--1513.

\bibitem[\protect\citeauthoryear{Edwards and Taylor}{Edwards and
  Taylor}{2009}]{edwards2009}
Edwards, H.~H. and M.~D. Taylor (2009).
\newblock Characterizations of degree one bivariate measures of concordance.
\newblock {\em J. Multivariate Anal.\/}~{\em 100}, 1777--1791.

\bibitem[\protect\citeauthoryear{Fermanian, Radulovic, and Wegkamp}{Fermanian
  et~al.}{2004}]{fermanian2004weak}
Fermanian, J.-D., D.~Radulovic, and M.~Wegkamp (2004).
\newblock Weak convergence of empirical copula processes.
\newblock {\em Bernoulli\/}~{\em 10\/}(5), 847--860.

\bibitem[\protect\citeauthoryear{Fredricks and Nelsen}{Fredricks and
  Nelsen}{2007}]{fredricks2007relationship}
Fredricks, G.~A. and R.~B. Nelsen (2007).
\newblock On the relationship between {S}pearman's rho and {K}endall's tau for
  pairs of continuous random variables.
\newblock {\em J. Statist. Plann. Inference\/}~{\em 137\/}(7), 2143--2150.

\bibitem[\protect\citeauthoryear{Fuchs}{Fuchs}{2014}]{fuchs2014multivariate}
Fuchs, S. (2014).
\newblock Multivariate copulas: Transformations, symmetry, order and measures
  of concordance.
\newblock {\em Kybernetika\/}~{\em 50\/}(5), 725--743.

\bibitem[\protect\citeauthoryear{Fuchs}{Fuchs}{2016a}]{fuchs2016biconvex}
Fuchs, S. (2016a).
\newblock A biconvex form for copulas.
\newblock {\em Depend. Model.\/}~{\em 4\/}(1), 63--75.

\bibitem[\protect\citeauthoryear{Fuchs}{Fuchs}{2016b}]{fuchs2016copula}
Fuchs, S. (2016b).
\newblock Copula--induced measures of concordance.
\newblock {\em Depend. Model.\/}~{\em 4\/}(1), 205--214.

\bibitem[\protect\citeauthoryear{Fuchs}{Fuchs}{2023}]{sfx2022phi}
Fuchs, S. (2023).
\newblock Quantifying directed dependence via dimension reduction.
\newblock {\em J. Multivariate Anal., to appear\/}.

\bibitem[\protect\citeauthoryear{Fuchs and Schmidt}{Fuchs and
  Schmidt}{2014}]{fuchs2014bivariate}
Fuchs, S. and K.~D. Schmidt (2014).
\newblock Bivariate copulas: Transformations, asymmetry and measures of
  concordance.
\newblock {\em Kybernetika\/}~{\em 50\/}(1), 109--125.

\bibitem[\protect\citeauthoryear{Fuchs and Tschimpke}{Fuchs and
  Tschimpke}{2023}]{fuchs2023total}
Fuchs, S. and M.~Tschimpke (2023).
\newblock Total positivity of copulas from a {M}arkov kernel perspective.
\newblock {\em J. Math. Anal. Appl.\/}~{\em 518\/}(1), Article ID 126629.

\bibitem[\protect\citeauthoryear{Genest, Ne{\v{s}}lehov{\'a}, and
  Ben~Ghorbal}{Genest et~al.}{2010}]{genest2010spearman}
Genest, C., J.~Ne{\v{s}}lehov{\'a}, and N.~Ben~Ghorbal (2010).
\newblock Spearman's footrule and {G}ini's gamma: A review with complements.
\newblock {\em J. Nonparametr. Stat.\/}~{\em 22\/}(8), 937--954.

\bibitem[\protect\citeauthoryear{Genest, Ne{\v{s}}lehov{\'a}, and
  R{\'e}millard}{Genest et~al.}{2017}]{genest2017asymptotic}
Genest, C., J.~Ne{\v{s}}lehov{\'a}, and B.~R{\'e}millard (2017).
\newblock Asymptotic behavior of the empirical multilinear copula process under
  broad conditions.
\newblock {\em J. Multivariate Anal.\/}~{\em 159}, 82--110.

\bibitem[\protect\citeauthoryear{Gudendorf and Segers}{Gudendorf and
  Segers}{2011}]{GS2011}
Gudendorf, G. and J.~Segers (2011).
\newblock Nonparametric estimation of an extreme-value copula in arbitrary
  dimensions.
\newblock {\em J. Multivariate Anal.\/}~{\em 102\/}(1), 37 -- 47.

\bibitem[\protect\citeauthoryear{Guillem}{Guillem}{2000}]{guillem2000structure}
Guillem, A. I.~G. (2000).
\newblock Structure de d{\'e}pendance des lois de valeurs extr{\^e}mes
  bivari{\'e}es.
\newblock {\em C. R. Acad. Sci.\/}~{\em 330\/}(7), 593--596.

\bibitem[\protect\citeauthoryear{H{\"u}rlimann}{H{\"u}rlimann}{2003}]{hurlimann2003hutchinson}
H{\"u}rlimann, W. (2003).
\newblock Hutchinson-{L}ai's conjecture for bivariate extreme value copulas.
\newblock {\em Stat. Prob. Lett.\/}~{\em 61\/}(2), 191--198.

\bibitem[\protect\citeauthoryear{Hutchinson and Lai}{Hutchinson and
  Lai}{1990}]{hutchinson1990}
Hutchinson, T.~P. and C.~Lai (1990).
\newblock {\em Continuous Bivariate Distributions, Emphasising Applications}.
\newblock Rumsby, Adelaide.

\bibitem[\protect\citeauthoryear{Janssen, Swanepoel, and Veraverbeke}{Janssen
  et~al.}{2012}]{janssen2012large}
Janssen, P., J.~Swanepoel, and N.~Veraverbeke (2012).
\newblock Large sample behavior of the {B}ernstein copula estimator.
\newblock {\em J. Statist. Plann. Inference\/}~{\em 142\/}(5), 1189--1197.

\bibitem[\protect\citeauthoryear{Joe}{Joe}{1990}]{joe1990}
Joe, H. (1990).
\newblock Multivariate concordance.
\newblock {\em J. Multivariate Anal.\/}~{\em 35}, 12--30.

\bibitem[\protect\citeauthoryear{Joe}{Joe}{2014}]{joe2014dependence}
Joe, H. (2014).
\newblock {\em Dependence Modeling with Copulas}.
\newblock CRC Press Boca Raton, FL.

\bibitem[\protect\citeauthoryear{Junker, Griessenberger, and Trutschnig}{Junker
  et~al.}{2020}]{fgwt2020}
Junker, R., F.~Griessenberger, and W.~Trutschnig (2020).
\newblock Estimating scale-invariant directed dependence of bivariate
  distributions.
\newblock {\em Comput. Statist. Data Anal.\/}~{\em 153}, Article ID 107058, 22
  pages.

\bibitem[\protect\citeauthoryear{Kallenberg}{Kallenberg}{2002}]{kallenberg2002}
Kallenberg, O. (2002).
\newblock {\em Foundations of {M}odern Probability}.
\newblock New York.

\bibitem[\protect\citeauthoryear{Karger, Conrad, B{\"o}hner, Kawohl, Kreft,
  Soria-Auza, Zimmermann, Linder, and Kessler}{Karger
  et~al.}{2017}]{karger2017}
Karger, D., O.~Conrad, J.~B{\"o}hner, T.~Kawohl, H.~Kreft, R.~Soria-Auza,
  N.~Zimmermann, H.~Linder, and M.~Kessler (2017).
\newblock Climatologies at high resolution for the earth's land surface areas.
\newblock {\em Sci. Data\/}~{\em 4}, Article ID 170122.

\bibitem[\protect\citeauthoryear{Karger, Conrad, B{\"o}hner, Kawohl, Kreft,
  Soria-Auza, Zimmermann, Linder, and Kessler}{Karger
  et~al.}{2018}]{karger2018}
Karger, D., O.~Conrad, J.~B{\"o}hner, T.~Kawohl, H.~Kreft, R.~Soria-Auza,
  N.~Zimmermann, H.~Linder, and M.~Kessler (2018).
\newblock Data from: Climatologies at high resolution for the earth's land
  surface areas [dataset].

\bibitem[\protect\citeauthoryear{Kasper, Fuchs, and Trutschnig}{Kasper
  et~al.}{2021}]{kasper2021weak}
Kasper, T.~M., S.~Fuchs, and W.~Trutschnig (2021).
\newblock On weak conditional convergence of bivariate {A}rchimedean and
  extreme value copulas, and consequences to nonparametric estimation.
\newblock {\em Bernoulli\/}~{\em 27\/}(4), 2217--2240.

\bibitem[\protect\citeauthoryear{Klenke}{Klenke}{2008}]{klenke2008}
Klenke, A. (2008).
\newblock {\em Wahrscheinlichkeitstheorie}.
\newblock Heidelberg.

\bibitem[\protect\citeauthoryear{Kruskal}{Kruskal}{1958}]{kruskal1958}
Kruskal, W.~H. (1958).
\newblock Ordinal measures of association.
\newblock {\em J. Amer. Statist. Assoc.\/}~{\em 53}, 814--861.

\bibitem[\protect\citeauthoryear{Li, Mikusi{\'n}ski, and Taylor}{Li
  et~al.}{1998}]{li1998strong}
Li, X., P.~Mikusi{\'n}ski, and M.~D. Taylor (1998).
\newblock Strong approximation of copulas.
\newblock {\em J. Math. Anal. Appl.\/}~{\em 225\/}(2), 608--623.

\bibitem[\protect\citeauthoryear{Marshall, Olkin, and Arnold}{Marshall
  et~al.}{2011}]{marshallolkin2011}
Marshall, A., I.~Olkin, and B.~Arnold (2011).
\newblock {\em Inequalities: Theory of Majorization and Its Applications\/} (2
  ed.).
\newblock New York: Springer.

\bibitem[\protect\citeauthoryear{Mroz, Fuchs, and Trutschnig}{Mroz
  et~al.}{2021}]{sfx2021vine}
Mroz, T., S.~Fuchs, and W.~Trutschnig (2021).
\newblock How simplifying and flexible is the simplifying assumption in
  pair-copula constructions – analytic answers in dimension three and a
  glimpse beyond.
\newblock {\em Electron. J. Statist.\/}~{\em 15\/}(1), 1951--1992.

\bibitem[\protect\citeauthoryear{Munroe, Ransford, and Genest}{Munroe
  et~al.}{2010}]{munroe2010}
Munroe, P., T.~Ransford, and C.~Genest (2010).
\newblock A counterexample to a conjecture of {H}utchinson and {L}ai.
\newblock {\em C. R. Math. Acad. Sci. Paris\/}~{\em 348}, 305--310.

\bibitem[\protect\citeauthoryear{Nelsen}{Nelsen}{2006}]{nelsen2007introduction}
Nelsen, R.~B. (2006).
\newblock {\em An Introduction to Copulas}.
\newblock Springer, New York.

\bibitem[\protect\citeauthoryear{P{\'e}rez and Prieto-Alaiz}{P{\'e}rez and
  Prieto-Alaiz}{2016}]{perez2016}
P{\'e}rez, A. and M.~Prieto-Alaiz (2016).
\newblock A note on nonparametric estimation of copula-based multivariate
  extensions of {S}pearman's rho.
\newblock {\em Stat. Prob. Lett.\/}~{\em 112}, 41--50.

\bibitem[\protect\citeauthoryear{Pickands}{Pickands}{1981}]{Pickands1981}
Pickands, J. (1981).
\newblock Multivariate extreme value distributions.
\newblock {\em Proceedings 43rd Session International Statistical
  Institute\/}~{\em 2}, 859--878.

\bibitem[\protect\citeauthoryear{R{\'e}millard and Scaillet}{R{\'e}millard and
  Scaillet}{2009}]{remillard2009}
R{\'e}millard, B. and O.~Scaillet (2009).
\newblock Testing for equality between two copulas.
\newblock {\em J. Multivariate Anal.\/}~{\em 100}, 377--386.

\bibitem[\protect\citeauthoryear{Scarsini}{Scarsini}{1984}]{sca1984}
Scarsini, M. (1984).
\newblock On measures of concordance.
\newblock {\em Stochastica\/}~{\em 8}, 201--218.

\bibitem[\protect\citeauthoryear{Schmid and Schmidt}{Schmid and
  Schmidt}{2006}]{schmid2006multivariate}
Schmid, F. and R.~Schmidt (2006).
\newblock Multivariate extensions of {S}pearman’s rho and related statistics.
\newblock {\em Stat. Prob. Lett.\/}~{\em 77\/}(4).

\bibitem[\protect\citeauthoryear{Schreyer, Paulin, and Trutschnig}{Schreyer
  et~al.}{2017}]{schreyer2017exact}
Schreyer, M., R.~Paulin, and W.~Trutschnig (2017).
\newblock On the exact region determined by {K}endall's $\tau$ and {S}pearman's
  $\rho$.
\newblock {\em J. R. Stat. Soc. Ser. B. Stat. Methodol.\/}, 613--633.

\bibitem[\protect\citeauthoryear{Segers}{Segers}{2012}]{segers2012asymptotics}
Segers, J. (2012).
\newblock Asymptotics of empirical copula processes under non-restrictive
  smoothness assumptions.
\newblock {\em Bernoulli\/}~{\em 18\/}(3).

\bibitem[\protect\citeauthoryear{Trutschnig and Mroz}{Trutschnig and
  Mroz}{2018}]{trutschnig2018sharp}
Trutschnig, W. and T.~Mroz (2018).
\newblock A sharp inequality for {K}endall’s $\tau$ and {S}pearman’s $\rho$
  of extreme-value copulas.
\newblock {\em Depend. Model.\/}~{\em 6\/}(1), 369--376.

\bibitem[\protect\citeauthoryear{Trutschnig, Schreyer, and
  Fern{\'a}ndez-S{\'a}nchez}{Trutschnig et~al.}{2016}]{trutschnig2016mass}
Trutschnig, W., M.~Schreyer, and J.~Fern{\'a}ndez-S{\'a}nchez (2016).
\newblock Mass distributions of two-dimensional extreme-value copulas and
  related results.
\newblock {\em Extremes\/}~{\em 19\/}(3), 405--427.

\bibitem[\protect\citeauthoryear{Van~der Vaart}{Van~der
  Vaart}{2000}]{van2000asymptotic}
Van~der Vaart, A.~W. (2000).
\newblock {\em Asymptotic Statistics}.
\newblock Cambridge University Press.

\bibitem[\protect\citeauthoryear{Van~der Vaart and Wellner}{Van~der Vaart and
  Wellner}{1996}]{VaartWellner1996}
Van~der Vaart, A.~W. and J.~A. Wellner (1996).
\newblock {\em Weak Convergence and Empirical Processes}.
\newblock Springer, New York.

\end{thebibliography}

\section{Appendix} \label{Sec.App}

\subsection{Proofs of Section \ref{Sec.PMI}}

\begin{proof} \textbf{(Example \ref{Cor.Gaussian})}: 
To prove that, for $\rho \in (0,1)$, the Gaussian copula $C_\rho$ is PMI, we verify Ineq. \eqref{PMI.density}.
Therefore, let $c_\rho$ denote the density of $C_\rho$ and define \begin{align*}
  \Tilde{c}_\rho (s,t)
  := \frac{1}{2\pi \sqrt{1-\rho^2}} \exp{\left(-\frac{s^2-2\rho st + t^2}{2(1-\rho^2)} \right)}\,. 
\end{align*}
Then $c_\rho$ equals
\begin{align*}
    c_\rho(u,v) 
    & = \frac{\Tilde{c}_\rho(\Phi^{-1}(u), \Phi^{-1}(v))}{\Phi'(\Phi^{-1}(u)) \, \Phi'(\Phi^{-1}(v))}
\end{align*}
for all $ (u,v) \in (0,1)^2$. 
According to \cite{joe2014dependence} the density $c_\rho$ fulfills  
$c_\rho(u,v) = c_\rho(1-u, 1-v)$ for all $(u,v) \in (0,1)^2$, 
so that \eqref{PMI.density} simplifies to 
\begin{align}\label{AppendixExGaussianEq1}
    0 \leq c_\rho(u,v) - c_\rho(u,1-v) = \frac{1}{\Phi'(\Phi^{-1}(u))} \left( \frac{\Tilde{c}_\rho(\Phi^{-1}(u), \Phi^{-1}(v))}{\Phi'(\Phi^{-1}(v))} - \frac{\Tilde{c}_\rho(\Phi^{-1}(u), \Phi^{-1}(1-v))}{\Phi'(\Phi^{-1}(1-v))} \right)
\end{align}
Since $\Phi'(x) = \frac{1}{\sqrt{2\pi}} e^{-\frac{x^2}{2}}$ is positive, 
\eqref{AppendixExGaussianEq1} is equivalent to
\begin{align}\label{AppendixExGaussianEq2}
    0 
    & \leq \Phi'(\Phi^{-1}(1-v)) \, \Tilde{c}_\rho(\Phi^{-1}(u), \Phi^{-1}(v)) 
           - \Phi'(\Phi^{-1}(v)) \, \Tilde{c}_\rho(\Phi^{-1}(u), \Phi^{-1}(1-v)) \notag
    \\
    & = 
    \frac{1}{\sqrt{2\pi} 2\pi \sqrt{1-\rho^2}} \exp\left( \frac{-\Phi^{-1}(u)^2 + 2\rho \, \Phi^{-1}(u)\Phi^{-1}(v) - \Phi^{-1}(v)^2 - (1-\rho^2) \, \Phi^{-1}(1-v)^2}{2(1-\rho^2)} \right) \notag
    \\
    &\hspace{0.5cm} - 
    \frac{1}{\sqrt{2\pi} 2\pi \sqrt{1-\rho^2}} \exp\left( \frac{-\Phi^{-1}(u)^2 + 2\rho \, \Phi^{-1}(u)\Phi^{-1}(1-v) - \Phi^{-1}(1-v)^2 - (1-\rho^2) \, \Phi^{-1}(v)^2}{2(1-\rho^2)} \right)
\end{align}
which then is, due to the fact that $\Phi^{-1}(v) = -\Phi^{-1}(1-v)$ for all $v \in (0,1)$, 
equivalent to 
\begin{align}\label{AppendixExGaussianEq4}
   2\rho \,\Phi^{-1}(u) \big( \Phi^{-1}(v) - \Phi^{-1}(1-v) \big)
   & = 4\rho \, \Phi^{-1}(u) \, \Phi^{-1}(v) 
   \geq 0
\end{align}
Finally, since $\Phi^{-1}$ is increasing and negative on $(0, \frac{1}{2})$,
Ineq. \eqref{AppendixExGaussianEq4} holds for all $(u,v) \in (0,\tfrac{1}{2})^2$ 
which proves Ineq. \eqref{PMI.density} and hence the assertion.
\end{proof}

\begin{proof} (\textbf{Example \ref{ExampleFrank}}):
To prove that, for $\delta \in (0,\infty)$, the Frank copula $C_\delta$ is PMI, we verify Ineq. \eqref{PMI.density}.
Therefore, let $c_\delta$ denote the density of $C_\delta$ given by
\begin{align*}
  c_\delta(u,v) = \frac{\delta(1 - e^{-\delta})e^{-\delta(u+v)}}{\left( 1 - e^{-\delta} - (1 - e^{-\delta u}) (1 - e^{-\delta v} )\right)^2}\,.
\end{align*}
According to \cite{joe2014dependence} the density $c_\delta$ fulfills  
$c_\delta(u,v) = c_\rho(1-u, 1-v)$ for all $(u,v) \in (0,1)^2$, 
so that \eqref{PMI.density} simplifies to 
\begin{align*}
  0 \leq c_\delta(u,v) - c_\delta(u,1-v)\,.
\end{align*}
To prove the latter define  
\begin{align*}
  \Delta 
  := \frac{c_\delta(u,v) - c_\delta(u,1-v) }{\delta (1-e^{-\delta}) e^{-\delta u}}
  & = \frac{e^{-\delta v}}{\left( 1 - e^{-\delta} - (1 - e^{-\delta u}) (1 - e^{-\delta v}) \right)^2}
      - \frac{ e^{-\delta(1-v)}}{\left( 1 - e^{-\delta} - (1 - e^{-\delta u}) (1 - e^{-\delta (1-v)}) \right)^2}
\end{align*}
and verify that 
\begin{align*}
  & \Delta \cdot \left( 1 - e^{-\delta} - (1 - e^{-\delta u}) (1 - e^{-\delta v}) \right)^2
           \cdot \left( 1 - e^{-\delta} - (1 - e^{-\delta u}) (1 - e^{-\delta (1-v)}) \right)^2
  \\
  & = e^{-\delta v} \left( 1 - e^{-\delta} - (1 - e^{-\delta u})( 1 - e^{-\delta (1-v)} )\right)^2 - e^{-\delta(1-v)} \left( 1 - e^{-\delta} - (1 - e^{-\delta u}) (1 - e^{-\delta v}) \right)^2 \\&=
  e^{-\delta v} \left( (1 - e^{-\delta})^2 - 2(1 - e^{-\delta}) (1 - e^{-\delta u})(1 - e^{-\delta (1-v)}) + (1 - e^{-\delta u})^2( 1 - e^{-\delta (1-v)})^2  \right) \\&\hspace{0.5cm}- e^{-\delta(1-v)} \left( (1 - e^{-\delta})^2 - 2 (1 - e^{-\delta}) (1 - e^{-\delta u}) (1 - e^{-\delta v}) + (1 - e^{-\delta u})^2 (1 - e^{-\delta v})^2 \right) \\&=
  (1 - e^{-\delta})^2 (e^{-\delta v} - e^{-\delta (1-v)}) - 2(1 - e^{-\delta}) (1 - e^{-\delta u}) \underbrace{\left((1 - e^{-\delta (1-v)}) e^{-\delta v} - (1 - e^{-\delta v}) e^{-\delta(1-v)}\right)}_{= e^{-\delta v} - e^{-\delta (1-v)}} \\&\hspace{1cm} + (1 - e^{-\delta u})^2 \underbrace{\left( (1 - e^{-\delta (1-v)})^2 e^{-\delta v} - (1 - e^{-\delta v})^2 e^{-\delta(1-v)} \right)}_{= e^{-\delta v} - e^{-\delta (1-v)} + e^{-\delta (2-v)} - e^{-\delta (v+1)} } \\&=
  (e^{-\delta v} - e^{-\delta (1-v)}) \left( (1 - e^{-\delta}) - (1 - e^{-\delta u}) \right)^2 + (1 - e^{-\delta u})^2 \left( e^{-\delta (2-v)} - e^{-\delta (v+1)}\right) \\&=
  (e^{-\delta v} - e^{-\delta (1-v)}) ( e^{-\delta u} - e^{-\delta} )^2 + (1 - e^{-\delta u})^2 e^{-\delta } (e^{-\delta (1-v)} - e^{-\delta v}) \\&= 
  (e^{-\delta v} - e^{-\delta (1-v)}) \left( (e^{-\delta u} - e^{-\delta} )^2 - (1 - e^{-\delta u})^2 e^{-\delta } \right) \\&=
  (e^{-\delta v} - e^{-\delta (1-v)}) (1 - e^{-\delta}) (e^{-2\delta u} - e^{-\delta})\,.
\end{align*}
Then, for all $(u,v) \in (0,\tfrac{1}{2})^2$,
\begin{align*}
  & c_\delta(u,v) - c_\delta(u,1-v) 
  \\
  & = \big( \delta (1-e^{-\delta}) e^{-\delta u} \big) \cdot \Delta
  \\
  & = \big( \delta (1-e^{-\delta}) e^{-\delta u} \big) \cdot 
      \frac{(e^{-\delta v} - e^{-\delta (1-v)}) (1 -    e^{-\delta}) (e^{-2\delta u} - e^{-\delta})}
      {\big( 1 - e^{-\delta} - (1 - e^{-\delta u}) (1 - e^{-\delta v}) \big)^2
           \cdot \big( 1 - e^{-\delta} - (1 - e^{-\delta u}) (1 - e^{-\delta (1-v)}) \big)^2}
  \geq 0
\end{align*}
which proves Ineq. \eqref{PMI.density} and hence the assertion.
\end{proof}

\begin{proof} (\textbf{Example \ref{ExampleEVC}}):
To prove that $C_A$ fails to be PMI, we apply Ineq. \eqref{PMI.MK}.
Therefore, let $h:(0,1)^2 \to \R$ given by $h(u,v) := \log(u) / \log(uv)$ and $F_A:[0,1) \to \R$ given by $F_A(t) := A(t) + (1-t) D^+A(t)$ with $D^+$ denoting the right-hand derivative of $A$. Then (a version of) the Markov kernel of $C_A$ (see \cite{trutschnig2016mass}) is given by 
\begin{align*} 
    K_{C_A}(u,[0,v]) = \begin{cases}
                        1 & \text{if } u \in \lbrace 0,1 \rbrace\\
                        \frac{C_A(u,v)}{u} F_A(h(u,v)) & \text{if } u,v \in (0,1) \\
                        v & \text{if } (u,v) \in (0,1) \times \lbrace 0,1 \rbrace
                    \end{cases}
\end{align*}
By setting $u \in [\exp(-0.99),\exp(-1.01)]$, $v_1 = \exp(-2)=0.1353$ and $v_2 = \exp(-1.5)=0.2231$ we observe
\begin{align*}
    & K_{C_A}(u,[0,v_1]) - K_{C_A}(1-u,[0,v_1]) + K_{C_A}(u,[0,1-v_1]) - K_{C_A}(1-u,[0,1-v_1]) 
    \\
    & > K_{C_A}(u,[0,v_2]) - K_{C_A}(1-u,[0,v_2]) + K_{C_A}(u,[0,1-v_2]) - K_{C_A}(1-u,[0,1-v_2])
\end{align*}
and therefore the mapping 
$$v \mapsto K_{C_A}(u,[0,v]) - K_{C_A}(1-u,[0,v]) + K_{C_A}(u,[0,1-v]) - K_{C_A}(1-u,[0,1-v])$$ 
fails to be non-decreasing on $(0,1/2)$, hence $C_A$ fails to be PMI.
\end{proof}

\subsection{Proofs of Section \ref{SectionConcordanceMeasureByCopula}}

In what follows we prove Theorem \ref{MainResultPMI.Thm}.
To do so, we first derive some useful representations of $4[C,B]-1$ for copulas $B$ being invariant.

\begin{lemma}\label{IntegralIdentity1}~~
\begin{enumerate}
\item For every copula $B$ satisfying $\nu_1(B) = B$ the identity
\begin{align*}
  4[C,B] - 1 = 2 \int_{\I^2} (C-\nu_1(C))(u,v) \, \mathrm{d} \mu_B (u,v)
\end{align*}
holds for all $C\in\C$.

\item For every copula $B$ satisfying $\nu_1(B) = \nu_2(B) = B$ the identity
\begin{align*}
  4[C,B] - 1 = \int_{\I^2} E_C(u,v) \, \mathrm{d} \mu_B (u,v)
\end{align*}
holds for all $C\in\C$.
\end{enumerate}
\end{lemma}

\begin{proof}
\cite[Corollary 5.5]{fuchs2016biconvex} gives 
$[\nu_1(C), B] = 1/2 - [C, \nu_1(B)] = 1/2 - [C, B]$ for all $C\in\C$,
which directly implies
\begin{align*}
    2[C,B] - \frac{1}{2}  
    = [C,B] - [\nu_1(C),B] 
    = \int_{\I^2} (C-\nu_1(C))(u,v) \, \mathrm{d} \mu_B (u,v)
\end{align*}
for all $C\in\C$, which proves the first assertion. \\
Now, for $(u,v)\in\I^2$ we have 
$(C-\nu_1(C))(1-u,1-v) = (\nu(C)-\nu_2(C))(u,v)$
which, together with \cite[Lemma 2.3]{fuchs2016biconvex}, implies
\begin{align*}
    4[C,B] - 1
    & = 2 \int_{\I^2} (C-\nu_1(C))(u,v) \, \mathrm{d} \mu_B(u,v)
    \\
    & = \int_{\I^2} (C-\nu_1(C))(u,v) \, \mathrm{d} \mu_B(u,v) + \int_{\I^2} (\nu(C)-\nu_2(C))(1-u,1-v) \, \mathrm{d} \mu_B(u,v)
    \\
    & = \int_{\I^2} (C-\nu_1(C))(u,v) \, \mathrm{d} \mu_B(u,v) + \int_{\I^2} (\nu(C)-\nu_2(C))(u,v) \, \mathrm{d} \mu_{\nu(B)}(u,v) 
    \\
    & = \int_{\I^2} (C-\nu_1(C))(u,v) \, \mathrm{d} \mu_B(u,v) + \int_{\I^2} (\nu(C)-\nu_2(C))(u,v) \, \mathrm{d} \mu_{B}(u,v)
    \\
    & = \int_{\I^2} E_C(u,v) \, \mathrm{d} \mu_{B}(u,v)
\end{align*}
for all $C\in\C$, which proves the second assertion.
\end{proof}

For convenience, define $\I_{i,j} := \left( \frac{i-1}{2}, \frac{i}{2} \right) \times \left( \frac{j-1}{2}, \frac{j}{2} \right)$, $i,j \in \{1,2\}$.

\begin{lemma}\label{IntegralIdentity2}
For every copula $B$ satisfying $\nu_1(B) = \nu_2(B) = B$ the identity 
\begin{align*}
    4[C,B] - 1 = 4 \, \int_{\I_{1,1}} E_C (u,v) \, \mathrm{d} \mu_B(u,v)
\end{align*}
holds for all $C\in\C$.
\end{lemma}

\begin{proof}
From Lemma \ref{IntegralIdentity1} we know that
$4[C,B] - 1 = \int_{\I^2} E_C(u,v) \, \mathrm{d} \mu_B(u,v)$
for all $C\in\C$.
Since $E_C$ equals 
$E_C (u,v) = C(u,v) + C(1-u,v) + C(u,1-v) + C(1-u,1-v) - 1$
and hence
$E_C (u,v) = E_C (1-u,v) = E_C (u,1-v) = E_C (1-u,1-v)$, 
from \cite[Lemma 2.3]{fuchs2016biconvex}
we obtain \pagebreak
\begin{align*}
    & 4[C,B] - 1 = \int_{\I^2} E_C(u,v) \, \mathrm{d} \mu_B(u,v)
    \\
    & = \int_{\I_{1,1}} E_C(u,v) \, \mathrm{d} \mu_B(u,v) + \int_{\I_{2,1}} E_C(u,v) \, \mathrm{d} \mu_B(u,v) 
        + \int_{\I_{1,2}} E_C(u,v) \, \mathrm{d} \mu_B(u,v) + \int_{\I_{2,2}} E_C(u,v) \, \mathrm{d} \mu_B(u,v)
    \\
    & = \int_{\I_{1,1}} E_C(u,v) \, \mathrm{d} \mu_B(u,v) 
        + \int_{\I_{1,1}} E_C(1-u,v) \, \mathrm{d} \mu_{\nu_1(B)}(u,v) 
		\\
		& \qquad  + \int_{\I_{1,1}} E_C(u,1-v) \, \mathrm{d} \mu_{\nu_2(B)}(u,v)  
                   + \int_{\I_{1,1}} E_C(1-u,1-v) \, \mathrm{d} \mu_{\nu(B)}(u,v) 
    \\
    & = \int_{\I_{1,1}} E_C(u,v) \, \mathrm{d} \mu_B(u,v) 
        + \int_{\I_{1,1}} E_C(u,v) \, \mathrm{d} \mu_{B}(u,v) 
        + \int_{\I_{1,1}} E_C(u,v) \, \mathrm{d} \mu_{B}(u,v)  
                   + \int_{\I_{1,1}} E_C(u,v) \, \mathrm{d} \mu_{B}(u,v) 
\end{align*}
where the last identity follows from the invariance of copula $B$.
This proves the assertion.
\end{proof}

Since $M_\Gamma$, $\Pi$ and $V$ are invariant, 
Lemma \ref{IntegralIdentity1} and Lemma \ref{IntegralIdentity2} yield the following representations for the measures of concordance 
Spearman's rho, Gini's gamma and $\kappa_V$.

\begin{corollary}\label{RepresentationRhoGamma}
Spearman's rho $\kappa_\Pi$, Gini's gamma $\kappa_{M_\Gamma}$ and $\kappa_V$ satisfy
\begin{align*}
    \kappa_\Pi(C) 
    & = 3 \int_{\I^2} E_C(u,v) \, \mathrm{d} \mu_{\Pi}(u,v) 
      = 12 \int_{\I_{1,1}} E_C(u,v) \, \mathrm{d} \mu_{\Pi}(u,v) 
    \\
    \kappa_{M_\Gamma}(C) 
    & = 2 \int_{\I^2} E_C(u,v) \, \mathrm{d} \mu_{M_\Gamma}(u,v) 
      = 8 \int_{\I_{1,1}} E_C(u,v) \, \mathrm{d} \mu_{M_\Gamma}(u,v) 
    \\
    \kappa_V(C) 
    & = 4 \int_{\I^2} E_C(u,v) \, \mathrm{d} \mu_{V}(u,v)
      = 16 \int_{\I_{1,1}} E_C(u,v) \, \mathrm{d} \mu_{V}(u,v)
\end{align*}
for all $C\in\C$.
\end{corollary}

\begin{proof}[Proof of Theorem \ref{MainResultPMI.Thm}]\label{ComparisonMainThm}~~
In order to prove Ineq. \eqref{Thm.PMI.Ineq}, we apply Lemma \ref{IntegralIdentity2} and an integration by parts formula for Lebesgue integrals (see, e.g., \cite{ansari2023}).
Since $A$ is invariant (hence $A(1/2,t) = t/2 = A(t,1/2)$ for all $t \in (0,1)$) and $E_C(0,v) = 0 = E_C(u,0)$ we have
\begin{align*}
    4[C,A] - 1 
    & = 4 \, \int_{\I_{1,1}} E_C (u,v) \, \mathrm{d} \mu_A(u,v)
    \\
    & = 4 \, \int_{\I_{1,1}}  A(1/2,1/2) - A(u,1/2) - A(1/2,v) + A(u,v) \, \mathrm{d} \mu_{E_C}(u,v)
    \\
    & = 4 \, \int_{\I_{1,1}}  A(u,v) \, \mathrm{d} \mu_{E_C}(u,v)
        + \int_{\I_{1,1}}  1 - 2u - 2v \, \mathrm{d} \mu_{E_C}(u,v)
\end{align*}
Therefore,
\begin{align*}
  \alpha(A) \kappa_B(C) - \alpha(B) \kappa_A(C)
  & = \alpha(A) \, \alpha(B) \; \Big( [C,B]-1/4 \big) - \alpha(B) \, \alpha(A) \; \Big( [C,A]-1/4 \big) 
  \\
  & = \alpha(A) \, \alpha(B) \;
      \int_{\I_{1,1}}  B(u,v) - A(u,v) \, \mathrm{d} \mu_{E_C}(u,v)
    \geq 0
\end{align*}
where the inequality is due to the fact that $A \preceq B$.
This proves the result.
\end{proof}

\subsection{Proofs of Section \ref{SectionEstimation}}

\begin{proof}[Proof of Lemma \ref{EstimaorMoCClassicalAndCheckerboard}]
First, recall that due to the absence of ties,
$\sum_{i=1}^n R_{ij} =  \sum_{i=1}^n i = \frac{n(n+1)}{2}$, $j \in \{1,2\}$. 
Integrating the empirical copula yields
\begin{align*}
  \int_{\I^2} C_n (u_1,u_2) \, \mathrm{d} \mu_A (u_1,u_2) 
  & = \frac{1}{n} \sum_{i=1}^n \int_{[0,1]^2} \prod_{j=1}^2 \mathds{1}_{[0, u_j]} \left( \frac{R_{ij}}{n+1}\right) \, \mathrm{d} \mu_A (u_1,u_2)  
  \\
  & = \frac{1}{n} \sum_{i=1}^n\int_{[0,1]^2} \mathds{1}_{\left[\tfrac{R_{i1}}{n+1}, 1\right] \times \left[\tfrac{R_{i2}}{n+1}, 1\right]} (u_1, u_2) \, \mathrm{d} \mu_A (u_1,u_2)  
  \\
  & = \frac{1}{n} \sum_{i=1}^n \mu_A \left( \left[\tfrac{R_{i1}}{n+1}, 1\right] \times \left[\tfrac{R_{i2}}{n+1}, 1\right] \right) 
  \\
  & = \frac{1}{n} \sum_{i=1}^n 
      \Big[ A(1,1) - A \left( \tfrac{R_{i1}}{n+1},1 \right) - A \left( 1,\tfrac{R_{i2}}{n+1}\right) + A \left( \tfrac{R_{i1}}{n+1}, \tfrac{R_{i2}}{n+1} \right) \Big] 
  \\
  & = \frac{1}{n} \sum_{i=1}^n 
      \Big[ 1 - \tfrac{R_{i1}}{n+1} - \tfrac{R_{i2}}{n+1} + A \left( \tfrac{R_{i1}}{n+1}, \tfrac{R_{i2}}{n+1} \right) \Big]
  \\
  & = \frac{1}{n} \sum_{i=1}^n A \left( \tfrac{R_{i1}}{n+1}, \tfrac{R_{i2}}{n+1} \right)\,.
\end{align*}
Now, considering the empirical checkerboard copula $\ec$ and using the fact that $\ec$ is a copula and absolutely continuous with density $\hat{c}_n$ that is piecewise constant on the interior of each rectangle 
$\big[\frac{k-1}{n}, \frac{k}{n}\big] \times \big[\frac{l-1}{n}, \frac{l}{n}\big]$,
$k,l \in \{1,\dots,n\}$, i.e.,
\begin{align*}
  \hat{c}_n(u_1,u_2)
  = \begin{cases}
      n \, \mathds{1}_{\big(\frac{k-1}{n}, \frac{k}{n}\big) \times \big(\frac{l-1}{n}, \frac{l}{n}\big)} (u_1,u_2) 
      & \textrm{if } \big(\tfrac{k}{n},\tfrac{l}{n}\big) = (F_{n,1} (X_{i1}), F_{n,2}(X_{i2}))
      \\
      0 
      & \textrm{otherwise}
    \end{cases},
\end{align*}
symmetry of the biconvex form (see, e.g., \cite[Theorem 3.3.]{fuchs2016biconvex}) yields
\begin{align*}
  \int_{\I^2} \ec(u_1,u_2) \, \mathrm{d} \mu_A(u_1,u_2) 
  & = \int_{\I^2} A(u_1,u_2) \, \hat{c}_n(u_1,u_2)  \, \mathrm{d}  \lambda(u_1,u_2)
  \\
	& = n \, \sum_{i=1}^n 
      \int_{\big(\frac{R_{i1}-1}{n},\frac{R_{i1}}{n}\big) \times \big(\frac{R_{i2}-1}{n},\frac{R_{i2}}{n}\big)} A(u_1,u_2) \, \mathrm{d}  \lambda(u_1,u_2) \,.
\end{align*}
If $C=M$, then the random variables $X_1$ and $X_2$ are comonotonic 
which implies
\begin{align*}
  \int_{\I^2} M_n (u_1,u_2) \, \mathrm{d} \mu_A (u_1,u_2) 
  & = \frac{1}{n} \sum_{i=1}^n A \left( \tfrac{R_{i1}}{n+1}, \tfrac{R_{i2}}{n+1} \right)
    = \frac{1}{n} \sum_{i=1}^n A \left(\frac{i}{n+1}, \frac{i}{n+1} \right) 
\end{align*}
and
\begin{align*}
  \int_{\I^2} \hat{M}_n (u_1,u_2) \, \mathrm{d} \mu_A(u_1,u_2)  
  & = n \, \sum_{i=1}^n \int_{\big(\frac{R_{i1}-1}{n},\frac{R_{i1}}{n}\big) \times \big(\frac{R_{i2}-1}{n},\frac{R_{i2}}{n}\big)} A(u_1,u_2) \, \mathrm{d}  \lambda(u_1,u_2)
  \\
  & = n \, \sum_{i=1}^n \int_{\big(\frac{i-1}{n}, \frac{i}{n}\big) \times \big(\frac{i-1}{n}, \frac{i}{n}\big)} A(u_1,u_2) \, \mathrm{d}  \lambda(u_1,u_2)\,.
\end{align*}
This proves the identities.
\\
Since $V$ is the least element with respect to the order $\preceq$ (see Proposition \ref{LemmaVarthetaProperties}) we have 
\begin{align*}
  \int_{\I^2} M_n (u_1,u_2) \, \mathrm{d} \mu_A (u_1,u_2) 
  &   =  \frac{1}{n} \sum_{i=1}^n A \left(\frac{i}{n+1}, \frac{i}{n+1} \right) 
    \geq \frac{1}{n} \sum_{i=1}^n V \left(\frac{i}{n+1}, \frac{i}{n+1} \right)
     >  \frac{1}{4}
\end{align*}
whenever $n \geq 4$ and, for $n \geq 2$, we obtain
\begin{align*}
  \int_{\I^2} \hat{M}_n (u_1,u_2) \, \mathrm{d} \mu_A(u_1,u_2) - \frac{1}{4}
  & =  \int_{\I^2} \hat{M}_n (u_1,u_2) - \Pi(u_1,u_2) \, \mathrm{d} \mu_A(u_1,u_2) 
  \\
	& =  \int_{\I^2} \int_{[0,u_1]\times [0,u_2]} h(s,t) 
         \, \mathrm{d} \lambda (s,t) \, \mathrm{d} \mu_A(u_1,u_2)
\end{align*}
where $h(s,t) = (n-1) \mathds{1}_{B_n} (s,t) - \mathds{1}_{\I^2 \backslash B_n} (s,t)$ with 
$B_n := \bigcup_{i=1}^n \big(\frac{i-1}{n}, \frac{i}{n}\big)^2$.
Since the inner integral 
$\int_{[0,u_1]\times [0,u_2]} h(s,t) \, \mathrm{d} \lambda (s,t) > 0$ for every $(u_1,u_2) \in (0,1)^2$, we can conclude that $\int_{\I^2} \hat{M}_n (u_1,u_2) \, \mathrm{d} \mu_A(u_1,u_2) > \frac{1}{4}$.
This proves the assertion.
\end{proof}

\begin{proof} (of Lemma \ref{AlphaConvergenceSpeed})
According to \cite{li1998strong} we have 
$d_\infty (\hat{M}_n,M) \leq \frac{2}{n}$ which directly yields 
\begin{align*}
    \vert \, [\hat{M}_n, A] - [M, A] \, \vert 
    \leq \int_{[0,1]^2} \vert \hat{M}_n (u_1,u_2) - M(u_1,u_2) \vert \, \mathrm{d} \mu_A(u_1,u_2) \leq \frac{2}{n}.
\end{align*}
Finally, since $[\hat{M}_n, A] > 1/4$ for every $n\geq 2$ 
(see Lemma \ref{EstimaorMoCClassicalAndCheckerboard}) and the mapping $x \mapsto (x-1/4)^{-1}$ is Lipschitz continuous on $(c,\infty)$ for every $c > 1/4$, 
we obtain  $\vert \hat{\alpha}_n(A) - \alpha(A)\vert = \mathcal{O} \left( \frac{1}{n} \right)$. 
\end{proof}

\begin{proof} (of Theorem \ref{Asymp.Norm.kappa})
Due to Lemma \ref{AlphaConvergenceSpeed} and Slutsky's theorem (see, e.g., \cite{van2000asymptotic}) we may first rewrite the two processes in the following way
\begin{align*}
    \sqrt{n} \, \big(\kappa_{A,n} - \kappa_A(C)\big) 
    & = \sqrt{n} \, \Big( \alpha_n(A) \, \big( [C_n,A] - \tfrac{1}{4} \big) - \alpha(A) \, \big([C,A] - \tfrac{1}{4}\big) \Big) 
    \\
    & = \sqrt{n} \, \big( [C_n,A] - [C,A] \big) \, \alpha_n(A) 
       + \sqrt{n} \, \big( \alpha_n(A) \, - \alpha(A) \big) \, [C,A] 
       - \sqrt{n} \; \frac{\alpha_n(A) - \alpha(A)}{4}
    \\
    & = \big[ \mathbb{C}_n,A \big] \, \underbrace{\alpha_n(A)}_{\overset{[\mathbb{P}]}{\to}\; \alpha(A)} 
       + \underbrace{\sqrt{n} \, \big( \alpha_n(A) \, - \alpha(A) \big)}_{\overset{[\mathbb{P}]}{\to}\; 0} \, [C,A] 
       - \underbrace{\sqrt{n} \; \frac{\alpha_n(A) - \alpha(A)}{4}}_{\overset{[\mathbb{P}]}{\to}\; 0}
\end{align*}
and analogously 
\begin{align*}
    \sqrt{n} \, \big(\hat{\kappa}_{A,n} - \kappa_A(C)\big) 
    & = \big[ \hat{\mathbb{C}}_n,A \big] \, \underbrace{\hat{\alpha}_n(A)}_{\overset{[\mathbb{P}]}{\to}\; \alpha(A)}  
       + \underbrace{\sqrt{n} \, \big( \hat{\alpha}_n(A) \, - \alpha(A) \big)}_{\overset{[\mathbb{P}]}{\to}\; 0} \, [C,A] 
       - \underbrace{\sqrt{n} \; \frac{\hat{\alpha}_n(A) - \alpha(A)}{4}}_{\overset{[\mathbb{P}]}{\to}\; 0}\,,
\end{align*}
so it remains to prove that the two sequences
$\big[ \mathbb{C}_n,A \big]$ and $\big[ \hat{\mathbb{C}}_n,A \big]$
converge weakly to $[\mathbb{C},A]$.
Since the map $l^\infty(\I^2) \to \mathbb{R}$ with $h \mapsto \int_{\I^2} h \, \mathrm{d} \mu_A$ is Hadamard differentiable at $C$ with derivative $l^\infty(\I^2) \to \mathbb{R}$ given by $h \mapsto \int_{\I^2} h \, \mathrm{d} \mu_A$, by the delta method (see, e.g., \cite[Theorem 3.9.4.]{VaartWellner1996}) 
and the asymptotic equivalence of the various copula estimators mentioned in Remark \ref{Emp.Cop.Remark},
$\big[ \mathbb{C}_n,A \big]$ and $\big[ \hat{\mathbb{C}}_n,A \big]$ converge weakly to
$[\mathbb{C},A] = \int_{\I^2} \mathbb{C}(\textbf{u}) \, \mathrm{d} \mu_A(\textbf{u})$.
\\
Finally, due to the fact that $\mathbb{C}$ is a centered Gaussian process and the map $l^\infty(\I^2) \to \mathbb{R}$ with $h \mapsto \int_{\I^2} h \, \mathrm{d} \mu_A$ is continuous and linear, the limit $\alpha(A) \, [\mathbb{C}, A]$ is centered Gaussian with variance given in \eqref{Asymp.Norm.kappa.Variance} (see, e.g., \cite[Section 3.9.2]{VaartWellner1996}).
\end{proof}

\begin{proof} (of Theorem \ref{Test.Thm})
We first prove that 
$\alpha_n(B) \sqrt{n} \, (\kappa_{A,n} - \kappa_A(C)) 
  - \alpha_n(A) \sqrt{n} \, (\kappa_{B,n} - \kappa_B(C))$
and $ \hat{\alpha}_n(B) \sqrt{n} \, (\hat{\kappa}_{A,n} - \kappa_A(C)) 
  - \hat{\alpha}_n(A) \sqrt{n} \, (\hat{\kappa}_{B,n} - \kappa_B(C))$
converge weakly to $\alpha(A) \alpha(B)\, \big([\mathbb{C},A] - [\mathbb{C},B]\big)$.
\\
Due to Lemma \ref{AlphaConvergenceSpeed} and Slutsky's theorem (see, e.g., \cite{van2000asymptotic}) we may first rewrite the process as follows
\begin{align*}
  & \alpha_n(B) \sqrt{n} \, (\kappa_{A,n} - \kappa_A(C)) 
  - \alpha_n(A) \sqrt{n} \, (\kappa_{B,n} - \kappa_B(C))
  \\
  & = \left( \alpha_n(B) \sqrt{n} \, \Big( \alpha_n(A) [C_n,A] - \alpha(A) [C,A] \Big) 
      - \frac{1}{4} \, \alpha_n(B) \sqrt{n} \, \Big( \alpha_n(A)  - \alpha(A) \Big) \right)
  \\
  & \qquad - \left( \alpha_n(A) \sqrt{n} \, \Big( \alpha_n(B) [C_n,B] - \alpha(B) [C,B] \Big) 
      - \frac{1}{4} \, \alpha_n(A) \sqrt{n} \, \Big( \alpha_n(B)  - \alpha(B) \Big) \right)
  \\
  & = \sqrt{n} \, \Big( \alpha_n(A)\alpha_n(B) \, [C_n,A] - \alpha(A)\alpha_n(B) \, [C,A] 
      - \alpha_n(A)\alpha_n(B) \, [C_n,B] + \alpha_n(A)\alpha(B) \, [C,B] \Big)
  \\
  & \qquad  
     + \frac{\sqrt{n}}{4} \, \Big( \alpha(A) \, \alpha_n(B)  - \alpha(B) \, \alpha_n(A) \Big)
  \\
  & = \underbrace{\alpha_n(A)\alpha_n(B)}_{\overset{[\mathbb{P}]}{\to} \; \alpha(A)\alpha(B)} \, \Big( [\mathbb{C}_n,A] - [\mathbb{C}_n,B] \Big)
   + \underbrace{\frac{\sqrt{n}}{4} \, \Big( \alpha(A) \, \alpha_n(B)  - \alpha(B) \, \alpha_n(A) \Big)}_{\overset{[\mathbb{P}]}{\to} 0}
  \\
  & \qquad  
     + \underbrace{\sqrt{n} \, \Big( \alpha_n(A) \alpha_n(B) - \alpha(A)\alpha_n(B) \Big)}_{\overset{[\mathbb{P}]}{\to} 0} \, [C,A]
     + \underbrace{\sqrt{n} \, \Big( \alpha_n(A)\alpha(B) - \alpha_n(A)\alpha_n(B) \Big)}_{\overset{[\mathbb{P}]}{\to} 0} \, [C,B]\,,    
\end{align*}
and analogously for $ \hat{\alpha}_n(B) \sqrt{n} \, (\hat{\kappa}_{A,n} - \kappa_A(C)) 
  - \hat{\alpha}_n(A) \sqrt{n} \, (\hat{\kappa}_{B,n} - \kappa_B(C))$.
So it remains to prove that the two sequences 
$[\mathbb{C}_n,A] - [\mathbb{C}_n,B]$ and 
$[\hat{\mathbb{C}}_n,A] - [\hat{\mathbb{C}}_n,B]$
converge weakly to $[\mathbb{C},A] - [\mathbb{C},B]$.
\\
Since the map $l^\infty(\I^2) \to \mathbb{R}$ with $h \mapsto \int_{\I^2} h \, \mathrm{d} \mu_A - \int_{\I^2} h \, \mathrm{d} \mu_B$ is Hadamard differentiable at $C$ with derivative $l^\infty(\I^2) \to \mathbb{R}$ given by $h \mapsto \int_{\I^2} h \, \mathrm{d} \mu_A - \int_{\I^2} h \, \mathrm{d} \mu_B$, by the delta method (see, e.g., \cite[Theorem 3.9.4.]{VaartWellner1996}) 
and the asymptotic equivalence of the various copula estimators mentioned in Remark \ref{Emp.Cop.Remark},
$[\mathbb{C}_n,A] - [\mathbb{C}_n,B]$ and $[\hat{\mathbb{C}}_n,A] - [\hat{\mathbb{C}}_n,B]$ converge weakly to
$[\mathbb{C},A] - [\mathbb{C},B]$.
Weak convergence of the original processes then results from using Slutsky's theorem a second time.
\\
Finally, due to the fact that $\mathbb{C}$ is a centered Gaussian process and the map $l^\infty(\I^2) \to \mathbb{R}$ with $h \mapsto \int_{\I^2} h \, \mathrm{d} Q^A - \int_{\I^2} h \, \mathrm{d} Q^B$ is continuous and linear, the limit $\alpha(A) \alpha(B) \, ([\mathbb{C}, A] - [\mathbb{C}, B])$ is centered Gaussian with variance given in \eqref{Asymp.Norm.Limit.Variance} (see, e.g., \cite[Section 3.9.2]{VaartWellner1996}).
\end{proof}

\end{document}